\documentclass[11pt,a4paper]{article}

\usepackage[margin=1in]{geometry}
\usepackage{amsmath,amssymb,amsthm}
\usepackage{hyperref}
\usepackage{cleveref}

\usepackage{comment}
\usepackage{enumitem}
\usepackage{braket}
\usepackage{graphicx}
\usepackage{booktabs}
\usepackage{xcolor}
\usepackage{orcidlink}

\usepackage{mathtools}
\newenvironment{redchange}{\begingroup}{\endgroup}
\newcommand{\rededit}[1]{{#1}}

\newtheorem{lemma}{Lemma}[section]
\newtheorem{proposition}[lemma]{Proposition}

\newtheorem{theorem}{Theorem}[section]

\newtheorem{corollary}[theorem]{Corollary}

\theoremstyle{definition}
\newtheorem{definition}[theorem]{Definition}

\theoremstyle{remark}

\DeclareMathOperator{\com}{com}
\DeclareMathOperator{\Var}{Var}
\newcommand{\HC}{H_C}
\newcommand{\HM}{H_M}
\newcommand{\Cavg}{C_{\text{avg}}}
\newcommand{\Cmax}{C_{\text{max}}}
\newcommand{\HMk}[1]{H_M^{(#1)}}

\newcommand{\Eaug}[1]{E_M^{(#1)}}
\newcommand{\dist}{\operatorname{dist}}
\newenvironment{restatedresult}[1]{%
  \par\medskip\noindent\textbf{#1 (restated from the main article).}\par\smallskip\itshape
}{\par\normalfont\medskip}

\title{Structural Bottlenecks in QAOA Mixer Design: Transport and Causal Reach}

\author{
  Jasjyot Singh Gulati
  \orcidlink{0009-0002-8425-355X}\\
  \small\textit{IIIT-Delhi, India}\\
  \small\texttt{jasjyot23257@iiitd.ac.in}
  \and
  Tharrmashastha SAPV\footnote{Corresponding Author}
  \orcidlink{0009-0002-8501-7663}\\
  \small\textit{EulerQ \& IIIT-Delhi, India}\\
  \small\texttt{tharrma@eulerq.com}
}
\date{}

\begin{document}
\maketitle

\begin{abstract}
The choice of mixer is central to the Quantum Approximate Optimization Algorithm (QAOA), yet no single structural measure captures every limitation of a shallow circuit. We study two complementary properties: configuration-space transport, which constrains probability concentration, and finite-depth causal reach, which limits the problem information available to local observables. In the transport analysis, first-order bounds in the small-$\beta$ and small-$\gamma$ regimes reduce to the same cost-weighted Dirichlet energy on connected regular mixer graphs. A non-perturbative $\ell_1$-spreading bound also gives a problem-dependent restriction on the expected objective and on the spreading required for a target approximation ratio. For MaxCut, we augment the transverse-field mixer with graph-aligned two-qubit interactions and establish a causal-radius bound of $(k+1)p$ at depth $p$. Numerical experiments at $k=1$ on cycles, circular ladders, and high-girth regular bipartite graphs show higher best-found approximation ratios for $X{+}XX$ and $X{+}YY$ than for standard $X$ at equal depth. Although these augmentations have the same causal-radius bound, $X{+}YY$ performs better in every tested setting, highlighting the role of operator structure alongside causal reach.

\end{abstract}

\textbf{Keywords:} 
Quantum Approximate Optimization Algorithm (QAOA), Mixer Hamiltonians, State-Space Exploration, Operator Lightcone Expansion, $L_1$-Spreading Norm, Perturbative Bounds, Mixer Augmentation

\section{Introduction}
\label{sec:intro}

Combinatorial optimization problems arise in many areas of science and engineering, often requiring the best solution to be found among an exponentially large number of possibilities. Solving such problems can be computationally difficult, particularly as their size and complexity increase. The Quantum Approximate Optimization Algorithm (QAOA)~\cite{farhi2014quantum} is a quantum approach to these problems that aims to find high-quality approximate solutions using parameterized quantum circuits. It alternates between two operations: a cost operator that encodes the optimization objective and a mixer that allows the algorithm to explore different candidate solutions. The performance of QAOA depends on how these operations work together, making it important to understand both its capabilities and its limitations.

A central question in QAOA is how the choice of mixer affects the quality of solutions obtained at a given circuit depth. Different mixers connect candidate solutions in different ways, leading to different patterns of quantum evolution and interference. Moreover, at shallow depths, a circuit may not be able to explore the solution space effectively or capture enough information about the problem. Understanding these limitations can help explain why certain mixers perform better than others and guide the design of more effective QAOA circuits.

In this work, we investigate two aspects of mixer design. The first is \emph{configuration-space transport}, which is related to how the mixer allows quantum amplitude to move between candidate solutions. A mixer that does not provide sufficient transport may limit the algorithm's ability to concentrate probability on high-quality solutions. The second is \emph{causal reach}, which describes how far information about the problem can spread through the circuit at a given depth. For problems with local interactions, such as MaxCut, this reach is limited by the interactions implemented in the QAOA circuit and the number of layers used. We study these two limitations separately to better understand their roles in shallow QAOA.

We investigate these limitations using two complementary approaches. First, we represent the mixer as a graph over feasible configurations and derive objective bounds in small-parameter and unrestricted-angle regimes. The perturbative results connect cost differences along mixer edges to a spectral description through the graph Dirichlet energy, while the non-perturbative results bound the probability available to high-value configurations. Second, we study MaxCut circuits with additional two-qubit mixer interactions and establish how those interactions change the causal reach of local cost terms.

We evaluate the resulting $k=1$ augmentations on cycles, circular ladders, and high-girth regular bipartite graphs using representative lightcones. Both $X{+}XX$ and $X{+}YY$ improve the best-found approximation ratios over standard $X$ at equal circuit depth in the tested settings. Their different performance despite identical causal-radius bounds motivates studying operator structure as well as interaction locality.

\section{Related Work}
\label{sec:related-work}

The choice of mixer has been studied extensively in QAOA. The original algorithm~\cite{farhi2014quantum} uses a transverse-field mixer, while the Quantum Alternating Operator Ansatz~\cite{hadfield2019quantum} provides a broader framework for designing mixers suited to particular problems and constraints. Several constructions have since been proposed, including constraint-preserving $XY$ mixers~\cite{fuchs2022constraint,rieffel2020xy}, Grover-based mixers~\cite{bartschi2020grover}, and adaptive mixer schemes~\cite{zhu2022adaptive}. Fuchs and Bassa~\cite{fuchs2024lx} developed a systematic approach to constructing mixers that preserve a given subspace. Other studies have explored ways to make QAOA circuits more expressive. Multi-angle QAOA assigns separate parameters to different terms of the circuit~\cite{herrman2022multiangle}, while Vijendran et al.~\cite{vijendran2024expressive} introduced XQAOA, which adds single-qubit Pauli-$Y$ rotations and additional variational parameters to improve performance at low depths. A broader overview of mixer design and QAOA variants can be found in~\cite{blekos2024review}.

Another line of research seeks to understand how mixer properties affect QAOA performance. Bridi and Marquezino~\cite{bridi2024grover} investigated Grover-based mixers and how their performance depends on the distribution of objective values. Ng et al.~\cite{ng2024analytical} derived analytical expressions for QAOA with product mixers and Grover-type mixers, highlighting differences in how these constructions depend on the problem graph. Recent work has also used spectral properties inspired by quantum walks to study mixer topologies and their relationship with QAOA performance~\cite{gulati2026spectral}. Onah et al.~\cite{onah2026separating} examined the distinction between amplitude transport and quantum interference, emphasizing that the ability to spread amplitude does not necessarily imply concentration on high-quality solutions. Our analysis follows a related motivation but focuses on deriving bounds on the QAOA objective from the connections between configurations and the spreading produced by mixer unitaries.

A separate line of works has studied the limitations of shallow QAOA circuits. For problems such as MaxCut, the expectation value of a local cost term depends only on a limited neighborhood of the problem graph at a fixed circuit depth~\cite{farhi2014quantum,farhi2020quantum}. Bravyi et al.~\cite{bravyi2020obstacles} established limitations arising from locality and symmetry in variational quantum optimization. The local structure of QAOA has also been used to analyze its performance on regular graphs. Wurtz and Love~\cite{wurtz2021maxcut} derived performance guarantees for shallow QAOA on 3-regular MaxCut instances, while Basso et al.~\cite{basso2022highdepth} developed methods for evaluating QAOA on large-girth regular graphs using their local tree-like structure. These studies show both the limitations imposed by circuit depth and the usefulness of analyzing local neighborhoods instead of entire graphs.

Several approaches have attempted to overcome locality limitations by modifying the interactions in QAOA circuits. The Phantom Edges ansatz~\cite{langfitt2024phantom} introduces additional weighted edges into the cost Hamiltonian, allowing shallow circuits to access more information about the problem graph. More recently, Kumar and Tahir~\cite{kumar2026dealing} studied additional mixer couplings that modify the interaction geometry and expand the region accessible at shallow depths. Our approach is related to these studies, but we restrict the added mixer interactions to edges already present in the MaxCut problem graph. We examine $XX$ and $YY$ augmentations with the same number of variational parameters and use local graph structure to evaluate their performance on selected graph families.

\begin{redchange}
These studies motivate the two directions developed below: objective bounds from mixer transport and the effects of graph-aligned interactions on causal reach.
\end{redchange}

\section{Contributions}
\label{sec:intro-contributions}

\begin{redchange}
The main contributions of this work combine analytical bounds on mixer transport with a study of causal reach and graph-aligned mixer augmentations.
\end{redchange}

\begin{enumerate}[leftmargin=*]

\item \textbf{\rededit{Bounds on QAOA performance from mixer structure.}}

We derive bounds on the QAOA objective in two parameter regimes: small mixing angles ($\beta$) and small cost angles ($\gamma$). Although the two analyses describe the mixer differently, we show that, for connected regular mixer graphs, their leading-order expressions are exactly equivalent. Both are determined by the Dirichlet energy of the cost function on the mixer graph, which measures how much the objective values differ between connected configurations.

We also bound the error introduced by the first-order approximations, allowing us to obtain rigorous bounds on the exact QAOA objective within the specified parameter ranges. \begin{redchange}
The bounds allow comparison of mixer structures in restricted parameter regimes without solving the full parameter optimization problem.
\end{redchange}

\item \textbf{\rededit{A necessary condition for achieving a target approximation ratio.}}

We extend the transport analysis beyond the small-parameter regime by studying how much a mixer unitary can spread quantum amplitude in each layer. Using the $\ell_1$ spreading of the mixer, we derive an upper bound on the probability of measuring any candidate solution at the end of the QAOA circuit. We then use this bound, together with the objective values of all candidate solutions, to determine an upper bound on the expected QAOA objective.

\begin{redchange}
Inverting this objective bound gives a spreading threshold for a target approximation ratio and, when the spreading per layer is bounded, a lower bound on circuit depth. This provides a parameter-independent test for ruling out specified performance targets.
\end{redchange}

\item \textbf{\rededit{Expanding the causal reach of shallow QAOA circuits.}}

We investigate whether modifying the mixer can increase the amount of problem information accessible to a shallow QAOA circuit. Focusing on MaxCut, we augment the standard transverse-field mixer with two-qubit interactions aligned with the edges of the problem graph. Building on earlier work~\cite{kumar2026dealing}, we consider a general $k$-hop construction and establish that its causal radius is at most $(k+1)p$ at circuit depth $p$, compared with $p$ for the standard mixer.

We focus on the simplest augmentation, $k=1$, which introduces interactions only along existing problem-graph edges. For cycles, circular ladders, and high-girth regular bipartite graphs, we show how the QAOA objective can be evaluated exactly at fixed parameters using a finite set of representative local neighborhoods. This allows us to study the behavior of these circuits on large structured graphs without simulating the entire system.

\item \textbf{\rededit{Evaluating the effect of two-qubit mixer interactions.}}

We compare two augmented mixers, $X{+}XX$ and $X{+}YY$, with the standard transverse-field mixer on several families of MaxCut instances. Both augmentations introduce only one additional shared parameter per QAOA layer. Our numerical results show improvements in the best-found approximation ratios over the standard mixer at equal circuit depth, with $X{+}YY$ consistently outperforming $X{+}XX$ in the cases studied.

We also compare the augmented circuits with deeper standard QAOA circuits having the same causal-radius bound. In the cases examined, $X{+}YY$ at depth $p$ achieves approximation ratios close to those of the standard mixer at depth $2p$, while using fewer variational parameters. \begin{redchange}
The contrast between the two augmentations isolates an effect of operator choice beyond their shared radius bound.
\end{redchange}

\end{enumerate}


\section{Preliminaries}
\label{sec:prelims}

\begin{redchange}
We introduce the QAOA formulation and notation, followed by the configuration-space mixer graphs used for transport bounds and the local MaxCut interactions used for the causal-reach analysis.
\end{redchange}

\subsection{Optimization Problems and QAOA}
\label{sec:prelims-qaoa}

We consider combinatorial optimization problems in which the goal is to find a bit string $x$ that maximizes an objective function $C(x)$. For a problem involving $n$ binary variables, let $F \subseteq \{0,1\}^n$ denote the set of feasible solutions. In many problems, the objective function can be expressed as a sum of terms,
\[
C(x) = \sum_{a=1}^{l} C_a(x),
\qquad x \in F,
\]
where each $C_a(x)$ depends on a subset of the variables. For unconstrained problems, every bit string is feasible, so $F=\{0,1\}^n$.

QAOA~\cite{farhi2014quantum} approaches this optimization problem using a parameterized quantum circuit. The objective function is represented by a cost Hamiltonian that is diagonal in the computational basis,
\[
H_C = \sum_{x\in F} C(x)\ket{x}\bra{x}.
\]
The corresponding cost unitary is
\[
U_C(\gamma)=e^{-i\gamma H_C},
\]
where $\gamma$ is a variational parameter. This operation applies phases to computational-basis states according to their objective values.

The second component is the mixer Hamiltonian $H_M$, which determines how different candidate solutions are coupled during the quantum evolution. Its corresponding unitary is
\[
U_M(\beta)=e^{-i\beta H_M},
\]
where $\beta$ is another variational parameter. By alternating the cost and mixer unitaries, QAOA uses the objective-dependent phases and the mixing operations to produce a probability distribution over candidate solutions.

Throughout the paper, we use the uniform superposition over feasible solutions as the initial state,
\[
\ket{\psi_0}
=
\frac{1}{\sqrt{|F|}}
\sum_{x\in F}\ket{x}.
\]
A depth-$p$ QAOA circuit applies $p$ alternating cost and mixer layers, with layer~$1$ acting first, giving
\[
\ket{\psi_p}
=
U_p\ket{\psi_0},
\qquad
U_p
=
U_M(\beta_p)U_C(\gamma_p)\cdots U_M(\beta_1)U_C(\gamma_1).
\]

The objective is to choose the parameters so that the expected value of the cost Hamiltonian is as large as possible. We denote the optimized depth-$p$ objective by
\[
M_p
=
\max_{\boldsymbol{\gamma},\boldsymbol{\beta}}
\bra{\psi_p}H_C\ket{\psi_p},
\]
where
$\boldsymbol{\gamma}=(\gamma_1,\ldots,\gamma_p)$
and
$\boldsymbol{\beta}=(\beta_1,\ldots,\beta_p)$.
The expectation value can be written explicitly as
\begin{equation}
\begin{aligned}
\bra{\psi_p} H_C \ket{\psi_p}
={}&
\bra{\psi_0}
e^{i\gamma_1 H_C}e^{i\beta_1 H_M}\cdots
e^{i\gamma_p H_C}e^{i\beta_p H_M}
\,H_C\,
\\
&\quad\times
e^{-i\beta_p H_M}e^{-i\gamma_p H_C}\cdots
e^{-i\beta_1 H_M}e^{-i\gamma_1 H_C}
\ket{\psi_0}.
\end{aligned}
\label{eq:qaoa-expectation}
\end{equation}
\begin{redchange}
This expression is the starting point for the transport bounds developed below.
\end{redchange}

\subsection{Mixer Graphs and Configuration Space}
\label{sec:prelims-mixer-graph}

One way to study a mixer is to represent it as a graph over the feasible solutions. In this representation, each vertex corresponds to a candidate solution, and an edge indicates that the mixer Hamiltonian directly couples the corresponding computational-basis states.

\rededit{For the transport analysis,} we consider mixers described by the adjacency matrix of an unweighted, undirected graph over $F$~\cite{fuchs2022constraint}. The mixer Hamiltonian is written as
\[
H_M
=
\sum_{x,y\in F}
(H_M)_{xy}\ket{x}\bra{y},
\]
where $(H_M)_{xy}=1$ if the configurations $x$ and $y$ are connected by an edge, and zero otherwise. Since the graph is undirected, $H_M$ is real and symmetric.

This graph describes which candidate solutions the mixer can connect directly. The corresponding unitary $e^{-i\beta H_M}$ allows amplitude to evolve through these connections, including through paths involving multiple edges. The structure of the graph therefore influences how amplitude spreads across the feasible solution space.

Some of our results require the mixer graph to be connected and regular. A connected graph has a path between every pair of vertices, while a $d$-regular graph has exactly $d$ edges incident on every vertex. These assumptions allow us to relate the mixer graph's spectral properties to differences in objective values between connected configurations.

\subsection{MaxCut and Local Interactions}
\label{sec:prelims-maxcut}

\rededit{For the causal-reach analysis,} we focus on the MaxCut problem. Given an undirected graph $G=(V,E)$, the MaxCut problem is to find a partition of the vertex set $V$ into two disjoint subsets that maximizes the number of edges with endpoints in different subsets.

We associate one qubit with each vertex and represent the objective using the cost Hamiltonian
\[
H_C
=
\sum_{(u,v)\in E}
\frac{I-Z_uZ_v}{2},
\]
where $Z_u$ denotes the Pauli-$Z$ operator acting on qubit $u$. Each term contributes one when the two vertices corresponding to an edge belongs to different subsets and zero otherwise.

For standard QAOA applied to MaxCut, the mixer Hamiltonian is
\[
H_M^{(0)}
=
\sum_{v\in V}X_v,
\]
where $X_v$ is the Pauli-$X$ operator acting on qubit $v$. The initial state is $\ket{+}^{\otimes n}$, which is the uniform superposition over all bit strings.

An important feature of this formulation is that the cost Hamiltonian contains two-qubit interactions only along the edges of the problem graph, while the standard mixer contains single-qubit operations. As a result, the information that can influence a local cost term is limited by the interactions in the QAOA circuit and its depth.
\begin{redchange}
This locality motivates the later study of two-qubit mixer interactions aligned with the problem graph.
\end{redchange}

\section{Perturbative Bounds from Mixer Structure}
\label{sec:bounds}

We begin by studying how the structure of a mixer affects the QAOA objective when the circuit parameters are small. In this regime, the expectation value can be approximated using the first-order terms of an operator expansion. These approximations allow us to relate changes in the objective to properties of the mixer graph, without directly solving the full parameter optimization problem.

We consider two cases: one in which the mixing angles $\beta_k$ are small, and another in which the cost angles $\gamma_k$ are small. \begin{redchange}
One expansion uses cost differences along mixer edges; the other uses the mixer eigenbasis. For connected regular mixer graphs, we show that the two lead to the same graph quantity.
\end{redchange}

Throughout this section, we restrict the parameters to
\[
0\leq\beta_k\leq\eta,
\qquad
0\leq\gamma_k\leq\delta,
\qquad k=1,\ldots,p,
\]
where $\eta,\delta>0$. Additional conditions on these bounds will be imposed for the two parameter regimes.

We also shift the cost function, when necessary, so that $C(x)\geq0$ for every $x\in F$, by replacing $C(x)$ with $C(x)-C_{\min}$. This changes the expected objective by a constant but does not affect the optimal solutions or optimizing circuit parameters. All quantities below refer to the shifted cost function unless stated otherwise.

\subsection{Nested-Commutator Expansions}
\label{sec:expansions}

Our analysis starts from the QAOA expectation value in \eqref{eq:qaoa-expectation}. We expand this expression in powers of either $\beta$ or $\gamma$, keeping terms up to first order. The higher-order terms will be bounded later in Section~\ref{sec:error-bounds}.
To simplify the expansion, we define a sequence of operators that describes how the cost Hamiltonian transforms under successive QAOA layers, peeling off the layers of $U_p$ from the last one inward. Starting with
\[
\mathcal{O}^{(p)}=\HC,
\]
we define, for $\ell=p,p-1,\ldots,1$,
\[
\mathcal{O}^{(\ell-1)}
=
e^{i\gamma_\ell\HC}e^{i\beta_\ell\HM}
\mathcal{O}^{(\ell)}
e^{-i\beta_\ell\HM}e^{-i\gamma_\ell\HC}.
\]
The exact depth-$p$ objective is then given as
\[
\bra{\psi_p}\HC\ket{\psi_p}
=
\bra{\psi_0}\mathcal{O}^{(0)}\ket{\psi_0}.
\]

We use the notation
\[
\com_A^n(B)
=
\underbrace{[A,[A,\ldots,[A,B]]]}_{n\text{ commutators}},
\qquad
\com_A^0(B)=B.
\]

\begin{lemma}[Expansion in $\beta$]
\label{lem:small-beta}
For each layer $\ell$, the evolved cost operator admits the exact expansion
\begin{equation}
\mathcal{O}^{(\ell-1)}
=
e^{i\gamma_\ell\HC}
\left(
\sum_{n=0}^{\infty}
\frac{(i\beta_\ell)^n}{n!}
\com_{\HM}^{n}(\mathcal{O}^{(\ell)})
\right)
e^{-i\gamma_\ell\HC}.
\label{eq:small-beta-expansion}
\end{equation}
Applying this expansion successively for $\ell=p,\ldots,1$ gives an expression for $\mathcal{O}^{(0)}$ in terms of the mixing angles and nested commutators.
\end{lemma}

\rededit{The proof is provided in Supplementary Material, Section~\ref{S-supp:perturbative}}

The same approach can be used to expand the objective in powers of the cost angles.

\begin{lemma}[Expansion in $\gamma$]
\label{lem:small-gamma}
For each layer $\ell$, the evolved cost operator follows the exact expansion
\begin{equation}
\mathcal{O}^{(\ell-1)}
=
\sum_{n=0}^{\infty}
\frac{(i\gamma_\ell)^n}{n!}
\com_{\HC}^{n}
\left(
e^{i\beta_\ell\HM}
\mathcal{O}^{(\ell)}
e^{-i\beta_\ell\HM}
\right).
\label{eq:small-gamma-expansion}
\end{equation}
Applying this expansion successively for $\ell=p,\ldots,1$ gives an expression for $\mathcal{O}^{(0)}$ in terms of the cost angles and nested commutators.

\end{lemma}

\rededit{The proof is provided in Supplementary Material, Section~\ref{S-supp:perturbative}}

\paragraph{First-order approximations.}

We now approximate the evolved cost operator by retaining only terms up to first order in the parameter being expanded.

We first consider the small-$\beta$ expansion, keeping the full dependence on the cost angles $\gamma_k$. From Hadamard's lemma,
\[
e^{i\beta_\ell\HM}
\mathcal{O}^{(\ell)}
e^{-i\beta_\ell\HM}
=
\mathcal{O}^{(\ell)}
+i\beta_\ell[\HM,\mathcal{O}^{(\ell)}]
+O(\beta_\ell^2).
\]

When all mixing angles vanish, the mixer unitaries become identities, and the recursive definition gives $\mathcal{O}^{(\ell)}=\HC$. Since the operator depends analytically on the mixing angles, its expansion around $\boldsymbol{\beta}=\mathbf{0}$ has the form
\[
\mathcal{O}^{(\ell)}
=
\HC+O(\|\boldsymbol{\beta}\|).
\]
Consequently, replacing $\mathcal{O}^{(\ell)}$ with $\HC$ inside the commutator introduces only second-order corrections.

Using the first-order expansion and
$\mathcal{O}^{(\ell)}=\HC+O(\|\boldsymbol{\beta}\|)$,
the recursive relation becomes
\[
\begin{aligned}
\mathcal{O}^{(\ell-1)}
={}&
e^{i\gamma_\ell\HC}
\mathcal{O}^{(\ell)}
e^{-i\gamma_\ell\HC}
\\
&+
i\beta_\ell
e^{i\gamma_\ell\HC}
[\HM,\HC]
e^{-i\gamma_\ell\HC}
+
O(\|\boldsymbol{\beta}\|^2).
\end{aligned}
\]

Applying this relation successively for $\ell=p,\ldots,1$, a first-order term introduced at layer $m$ is conjugated by the cost unitaries of layers $m,m-1,\ldots,1$. Since these unitaries commute, we have
\[
\begin{aligned}
&e^{i\gamma_1\HC}\cdots e^{i\gamma_m\HC}
[\HM,\HC]
e^{-i\gamma_m\HC}\cdots e^{-i\gamma_1\HC}
\\
&\qquad=
e^{i\Gamma_m\HC}
[\HM,\HC]
e^{-i\Gamma_m\HC},
\end{aligned}
\]
where
\[
\Gamma_m=\sum_{k=1}^{m}\gamma_k
\]
is the total cost phase accumulated by the state when the $m$th mixer acts.
Collecting the first-order contributions from all layers, we obtain
\begin{equation}
A_p
=
\HC
+
i\sum_{m=1}^{p}
\beta_m
e^{i\Gamma_m\HC}
[\HM,\HC]
e^{-i\Gamma_m\HC}.
\label{eq:first-order-operator}
\end{equation}
Thus, $A_p$ approximates the evolved cost operator to first order in the mixing angles while retaining the full dependence on the cost angles.

We next consider the small-$\gamma$ expansion, keeping the full dependence on the mixing angles. In this case, we expand the conjugation by the cost unitary:
\[
e^{i\gamma_\ell\HC}X e^{-i\gamma_\ell\HC}
=
X+i\gamma_\ell[\HC,X]+O(\gamma_\ell^2).
\]
When all cost angles are zero, the mixing unitaries combine, giving
\[
\mathcal{O}^{(m-1)}
=
e^{i\Sigma_m\HM}
\HC
e^{-i\Sigma_m\HM},
\qquad
\Sigma_m=\sum_{k=m}^{p}\beta_k,
\]
where $\Sigma_m$ is the total mixing angle applied after the $m$th cost layer.
This is the zeroth-order expression used when evaluating the first-order commutators. A contribution introduced by $\gamma_m$ is then conjugated by the mixing unitaries of the remaining layers $m-1,\ldots,1$, whose combined angle is $\Sigma_1-\Sigma_m$.

Collecting all first-order contributions gives
\begin{equation}
\begin{aligned}
B_p
={}&
e^{i\Sigma_1\HM}
\HC
e^{-i\Sigma_1\HM}
\\
&+
i\sum_{m=1}^{p}\gamma_m
e^{i(\Sigma_1-\Sigma_m)\HM}
\left[
\HC,\,
e^{i\Sigma_m\HM}
\HC
e^{-i\Sigma_m\HM}
\right]
e^{-i(\Sigma_1-\Sigma_m)\HM}.
\end{aligned}
\label{eq:first-order-gamma}
\end{equation}

The two approximations therefore describe the QAOA objective in different parameter regimes. The small-$\beta$ expansion retains the full dependence on the cost angles, while the small-$\gamma$ expansion retains the full dependence on the mixing angles.
\paragraph{Expectation values.}

We next evaluate these approximations in the uniform initial state
\[
\ket{\psi_0}
=
\frac{1}{\sqrt{|F|}}\sum_{x\in F}\ket{x}.
\]

Using the uniform initial state and the fact that $\HC$ is diagonal in the computational basis, we have
\[
\begin{aligned}
&\bra{\psi_0}
e^{i\Gamma_m\HC}
[\HM,\HC]
e^{-i\Gamma_m\HC}
\ket{\psi_0}
\\
&\quad=
\frac{1}{|F|}
\sum_{x,y\in F}
(\HM)_{xy}
\bigl(C(y)-C(x)\bigr)
e^{i\Gamma_m(C(x)-C(y))},
\end{aligned}
\]
where we used
\[
\bra{x}[\HM,\HC]\ket{y}
=
(\HM)_{xy}\bigl(C(y)-C(x)\bigr).
\]

Since $\HM$ is symmetric, the contributions from $(x,y)$ and $(y,x)$ can be combined. Writing $\Delta_{xy}=C(y)-C(x)$ and using
\[
e^{-i\theta}-e^{i\theta}=-2i\sin\theta,
\]
we obtain
\[
\begin{aligned}
&\bra{\psi_0}
e^{i\Gamma_m\HC}
[\HM,\HC]
e^{-i\Gamma_m\HC}
\ket{\psi_0}
\\
&\quad=
-\frac{2i}{|F|}
\sum_{\substack{x,y\in F\\x<y}}
(\HM)_{xy}\Delta_{xy}
\sin(\Gamma_m\Delta_{xy}).
\end{aligned}
\]

Substituting this expression into the definition of $A_p$ gives
\begin{equation}
\begin{aligned}
\bra{\psi_0}A_p\ket{\psi_0}
={}&
\bar C
+
\frac{2}{|F|}
\sum_{m=1}^{p}\beta_m
\sum_{\substack{x,y\in F\\x<y}}
(\HM)_{xy}
\bigl(C(y)-C(x)\bigr)
\\
&\hspace{12mm}\times
\sin\!\left(
\Gamma_m\bigl(C(y)-C(x)\bigr)
\right).
\end{aligned}
\label{eq:first-order-expectation}
\end{equation}
where
\[
\bar C
=
\frac{1}{|F|}
\sum_{x\in F}C(x)
\]
is the average objective value over the feasible configurations.
This expression shows how the first-order change in the objective depends on the cost differences between configurations connected by the mixer.

For the small-$\gamma$ analysis, we assume that the mixer graph is connected and $d$-regular. Let
\[
\HM
=
\sum_k\lambda_k\ket{\phi_k}\bra{\phi_k}
\]
be its eigendecomposition, with real orthonormal eigenvectors. The largest eigenvalue is $\lambda_0=d$, with corresponding eigenvector
\[
\ket{\phi_0}=\ket{\psi_0}.
\]
We also define
\[
C_{ab}
=
\bra{\phi_a}\HC\ket{\phi_b}.
\]

Since
\[
e^{iq\HM}\ket{\psi_0}
=
e^{iqd}\ket{\psi_0},
\]
a calculation in the eigenbasis of $\HM$ gives
\begin{equation}
\begin{aligned}
&\bra{\psi_0}
\left[
\HC,\,
e^{iq\HM}\HC e^{-iq\HM}
\right]
\ket{\psi_0}
\\
&\qquad=
2i\sum_{b\ne0}
|C_{0b}|^2
\sin\!\left(q(\lambda_b-\lambda_0)\right).
\end{aligned}
\end{equation}

The coefficients $C_{0b}$ describe how the cost
Hamiltonian couples the uniform state to the other
eigenmodes of the mixer. Using the completeness of
the eigenbasis, we have
\[
\begin{aligned}
\sum_b |C_{0b}|^2
&=
\bra{\psi_0}\HC
\left(\sum_b\ket{\phi_b}\bra{\phi_b}\right)
\HC\ket{\psi_0}
\\
&=
\bra{\psi_0}\HC^2\ket{\psi_0}.
\end{aligned}
\]
Since
$C_{00}=\bra{\psi_0}\HC\ket{\psi_0}=\bar C$,
excluding the uniform eigenmode gives
\[
\begin{aligned}
\sum_{b\ne0}|C_{0b}|^2
&=
\bra{\psi_0}\HC^2\ket{\psi_0}-\bar C^2
&=
\Var(C),
\end{aligned}
\]
where $\Var(C)$ is the variance of the objective
values under the uniform distribution on $F$.

Using \eqref{eq:first-order-gamma} and taking the expectation, we obtain
\begin{equation}
\bra{\psi_0}B_p\ket{\psi_0}
=
\bar C
+
2\sum_{m=1}^{p}\gamma_m
\sum_{b\ne0}
|C_{0b}|^2
\sin\!\left(
\Sigma_m(\lambda_0-\lambda_b)
\right).
\label{eq:gamma-expectation}
\end{equation}

We now have two first-order expressions for the QAOA objective: one based on cost differences across mixer edges and another based on the spectrum of the mixer Hamiltonian. In the following subsections, we use them to obtain upper and lower bounds.

\subsection{Small-\texorpdfstring{$\beta$}{beta} Bound}
\label{sec:small-beta-derivation}

We first consider the regime in which the mixing angles are small. We restrict the parameters to
\begin{equation}
0\leq\beta_i\leq\eta,
\qquad
0\leq\gamma_k\leq\delta,
\qquad
\delta
\leq
\frac{\pi}{2p(C_{\max}-C_{\min})}.
\label{eq:delta-condition}
\end{equation}

The restriction on $\delta$ ensures that the arguments of the sine functions in \eqref{eq:first-order-expectation} have magnitude at most $\pi/2$. Within this range, each contribution to the first-order objective is non-decreasing in the allowed parameters.

Consequently, the first-order approximation is maximized when all parameters take their largest allowed values:
\[
\beta_i=\eta,
\qquad
\gamma_i=\delta.
\]
At this point,
\[
\Gamma_m=m\delta.
\]

Using Jordan's inequality,
\[
\frac{2}{\pi}t\leq\sin t\leq t,
\qquad 0\leq t\leq\frac{\pi}{2},
\]
we obtain the following bounds.

\begin{proposition}[First-order small-$\beta$ bound]
\label{prop:linear-bounds}
Under the conditions of \eqref{eq:delta-condition},
\begin{equation}
\begin{aligned}
\bar C
+
\frac{2\eta\delta\,p(p+1)}{\pi|F|}
\mathcal{E}_M
&\leq
\max_{\substack{0\leq\beta_i\leq\eta\\
0\leq\gamma_i\leq\delta}}
\bra{\psi_0}A_p\ket{\psi_0}
\\
&\leq
\bar C
+
\frac{\eta\delta\,p(p+1)}{|F|}
\mathcal{E}_M,
\end{aligned}
\label{eq:linear-bounds}
\end{equation}
where
\[
\mathcal{E}_M
=
\sum_{\{x,y\}\in E_{\HM}}
\bigl(C(x)-C(y)\bigr)^2,
\]
and $E_{\HM}$ is the set of undirected edges of the mixer graph.
\end{proposition}

\rededit{The proof is provided in Supplementary Material, Section~\ref{S-supp:perturbative}}

The quantity $\mathcal{E}_M$ is the Dirichlet energy of the cost function on the mixer graph. It measures the variation in objective values between configurations that are directly connected by the mixer.

The proposition shows that, within the first-order approximation and the specified parameter range, the improvement over the average objective is bounded by expressions proportional to this quantity.

\subsection{Small-\texorpdfstring{$\gamma$}{gamma} Bound}
\label{sec:small-gamma-derivation}

We next consider the complementary regime, where the cost angles are small. In this case, the first-order objective is expressed using the eigenvalues of the mixer Hamiltonian and their overlaps with the cost function.

We impose
\begin{equation}
0\leq\gamma_k\leq\delta,
\qquad
0\leq\beta_k\leq\eta,
\qquad
\eta
\leq
\frac{\pi}{2p(\lambda_0-\lambda_{\min})},
\label{eq:eta-condition}
\end{equation}
where $\lambda_{\min}$ is the smallest eigenvalue of $\HM$.

This condition ensures that the sine functions in \eqref{eq:gamma-expectation} remain in their increasing range. As in the small-$\beta$ case, the first-order expression is therefore maximized at
\[
\beta_i=\eta,
\qquad
\gamma_i=\delta,
\]
with
\[
\Sigma_m=(p-m+1)\eta.
\]

Applying Jordan's inequality gives the corresponding spectral bounds.

\begin{proposition}[First-order small-$\gamma$ bound]
\label{prop:gamma-bound}
Under the conditions of \eqref{eq:eta-condition},
\begin{equation}
\begin{aligned}
\bar C
+
\frac{2\delta\eta\,p(p+1)}{\pi}
\sum_{b\ne0}
(\lambda_0-\lambda_b)|C_{0b}|^2
&\leq
\max_{\substack{0\leq\beta_i\leq\eta\\
0\leq\gamma_i\leq\delta}}
\bra{\psi_0}B_p\ket{\psi_0}
\\
&\leq
\bar C
+
\delta\eta\,p(p+1)
\sum_{b\ne0}
(\lambda_0-\lambda_b)|C_{0b}|^2.
\end{aligned}
\label{eq:gamma-bound}
\end{equation}
\end{proposition}

\rededit{The proof is provided in Supplementary Material, Section~\ref{S-supp:perturbative}}

Unlike the previous result, these bounds are written in terms of the mixer spectrum. Each contribution depends on the distance of an eigenvalue from $\lambda_0$, weighted by the corresponding overlap with the cost Hamiltonian.

A simpler bound can be obtained using only the spectral gap and the spectral spread. Let
$\Delta=\lambda_0-\lambda_1$
denote the spectral gap, where $\lambda_1$ is the second-largest eigenvalue.

Since
\[
\Delta
\leq\lambda_0-\lambda_b
\leq\lambda_0-\lambda_{\min}
\]
for all $b\ne0$, we obtain the following result.

\begin{corollary}
\label{cor:gamma-coarse}
Under the conditions of \eqref{eq:eta-condition},
\begin{equation}
\begin{aligned}
\bar C
+
\frac{2\Delta\,\delta\eta\,p(p+1)}{\pi}
\Var(C)
&\leq
\max_{\substack{0\leq\beta_i\leq\eta\\
0\leq\gamma_i\leq\delta}}
\bra{\psi_0}B_p\ket{\psi_0}
\\
&\leq
\bar C
+
(\lambda_0-\lambda_{\min})
\delta\eta\,p(p+1)\Var(C).
\end{aligned}
\label{eq:gamma-edge-set-dominates}
\end{equation}
\end{corollary}

These bounds are less detailed than those in \Cref{prop:gamma-bound}, but require only the extreme spectral values and the variance of the cost function.

\paragraph{Equivalence of the two descriptions.}

We have obtained two apparently different bounds: one based on differences in objective values across mixer edges, and another based on the mixer spectrum. For connected regular mixer graphs, however, the quantities appearing in these bounds are exactly equal.

\begin{proposition}[Equivalence of the perturbative regimes]
\label{prop:equivalence}
Let $\HM$ be the adjacency matrix of a connected $d$-regular graph on $F$. Then
\begin{equation}
\mathcal{E}_M
=
|F|
\sum_{b\ne0}
(\lambda_0-\lambda_b)|C_{0b}|^2.
\label{eq:equivalence-identity}
\end{equation}
Consequently, the leading-order upper and lower expressions in \eqref{eq:linear-bounds} and \eqref{eq:gamma-bound} are identical after this substitution. The exact-objective bounds obtained after including truncation errors may have different remainder terms.
\end{proposition}

\rededit{The proof is provided in Supplementary Material, Section~\ref{S-supp:perturbative}}

\begin{redchange}
The identity in \Cref{prop:equivalence} allows the leading-order bounds to be interpreted either through mixer edges or through cost overlaps with mixer eigenmodes. Error estimates are needed before these first-order expressions can bound the exact objective.
\end{redchange}

\subsection{Truncation Error}
\label{sec:error-bounds}

The bounds derived so far apply to first-order approximations of the QAOA objective. We now estimate the error from discarding higher-order terms. This allows us to extend the results to the exact objective within the chosen parameter ranges.

Let
\[
R_\beta=\mathcal{O}^{(0)}-A_p,
\qquad
R_\gamma=\mathcal{O}^{(0)}-B_p
\]
denote the operator remainders of the two expansions.

\begin{redchange}
By bounding the discarded nested commutators and summing the terms of order at least two, we obtain
\begin{equation}
\begin{aligned}
\|R_\beta\|&\leq 4\|\HM\|^2\left(\sum_{\ell=1}^p|\beta_\ell|\right)^2\|\HC\|,\\
\|R_\gamma\|&\leq 4\|\HC\|^2\left(\sum_{\ell=1}^p|\gamma_\ell|\right)^2\|\HC\|.
\end{aligned}
\label{eq:both-error-bounds}
\end{equation}
These bounds hold when the corresponding quantities $2\|\HM\|\sum_\ell|\beta_\ell|$ and $2\|\HC\|\sum_\ell|\gamma_\ell|$ are at most one. The derivation appears in Supplementary Section~\ref{S-supp:remainder}
\end{redchange}

The following corollaries give sufficient parameter restrictions for keeping these errors below a desired tolerance.

\begin{corollary}
\label{cor:eta-bound}
To ensure $\|R_\beta\|\leq\epsilon$, it suffices to choose
\[
\eta
\leq
\frac{1}{2p\|\HM\|}
\sqrt{\frac{\epsilon}{\|\HC\|}},
\]
provided that $2p\|\HM\|\eta\leq1$ also holds.
\end{corollary}

\begin{corollary}
\label{cor:delta-bound}
To ensure $\|R_\gamma\|\leq\epsilon$, it suffices to choose
\[
\delta
\leq
\frac{1}{2p\|\HC\|}
\sqrt{\frac{\epsilon}{\|\HC\|}},
\]
provided that $2p\|\HC\|\delta\leq1$ also holds.
\end{corollary}

\paragraph{Bounds on the exact objective.}

We can now combine the first-order results with the truncation-error bounds.

Define the optimized objective over the restricted parameter range as
\[
M_p(\eta,\delta)
=
\max_{\substack{0\leq\beta_i\leq\eta\\
0\leq\gamma_i\leq\delta}}
\bra{\psi_p}\HC\ket{\psi_p}.
\]
Since the difference between the exact objective and its first-order approximation is uniformly bounded, we can obtain corresponding bounds on $M_p(\eta,\delta)$.

\begin{corollary}[Exact small-$\beta$ objective bound]
\label{cor:exact-beta}
Assume the conditions of \Cref{prop:linear-bounds} and
\[
2p\|\HM\|\eta\leq1.
\]
Define
\[
\varepsilon_\beta
=
4p^2\eta^2\|\HM\|^2\|\HC\|.
\]
Then
\[
\begin{aligned}
\bar C
+
\frac{2\eta\delta\,p(p+1)}{\pi|F|}
\mathcal{E}_M
-\varepsilon_\beta
&\leq M_p(\eta,\delta)
\\
&\leq
\bar C
+
\frac{\eta\delta\,p(p+1)}{|F|}
\mathcal{E}_M
+\varepsilon_\beta.
\end{aligned}
\]
\end{corollary}

\rededit{The proof is provided in Supplementary Material, Section~\ref{S-supp:exact}}

\begin{corollary}[Exact small-$\gamma$ objective bound]
\label{cor:exact-gamma}
Assume the conditions of \Cref{prop:gamma-bound} and
\[
2p\|\HC\|\delta\leq1.
\]
Define
\[
\varepsilon_\gamma
=
4p^2\delta^2\|\HC\|^3.
\]
Then
\[
\begin{aligned}
\bar C
+
\frac{2\delta\eta\,p(p+1)}{\pi}
\sum_{b\ne0}
(\lambda_0-\lambda_b)|C_{0b}|^2
-\varepsilon_\gamma
&\leq
M_p(\eta,\delta)
\\
&\leq
\bar C
+
\delta\eta\,p(p+1)
\sum_{b\ne0}
(\lambda_0-\lambda_b)|C_{0b}|^2
+\varepsilon_\gamma.
\end{aligned}
\]
\end{corollary}

\rededit{The proof is provided in Supplementary Material, Section~\ref{S-supp:exact}}

These results establish bounds on the exact QAOA objective in the two restricted parameter regimes. Their usefulness depends on the size of the truncation error relative to the leading-order contribution, which is proportional to
\[
\eta\delta\,p^2\frac{\mathcal{E}_M}{|F|}.
\]
When the error is much smaller than this contribution, the Dirichlet energy provides useful information about the influence of mixer structure on the objective. When the error is comparable to or larger than the leading-order term, the first-order analysis becomes less informative.


\section{Non-Perturbative Bounds from Mixer Spreading}
\label{sec:nonperturbative}

\begin{redchange}
The bounds above rely on small mixing or cost angles. We now develop a complementary analysis that applies at arbitrary parameter values. Rather than expanding the objective, we bound the probability that the circuit assigns to each configuration and use this restriction to limit the expected objective. This yields a necessary spreading threshold for a target approximation ratio and a corresponding depth bound when individual mixer layers have limited spreading.
\end{redchange}

\subsection{Mixer Spreading and Output Probabilities}

For each layer $j$, define the matrix elements of the mixer unitary by
\[
K_j(a,b)
=
\bra{a}e^{-i\beta_j\HM}\ket{b}.
\]
Throughout this section, we assume that $\HM$ is real and symmetric. Consequently, the mixer unitary is also symmetric, so $K_j(a,b)=K_j(b,a)$. In particular, its row and column $\ell_1$-norms coincide, which allows us to characterize the spreading of a layer using either one.

We define the \emph{spreading} of the $j$th mixer layer as
\[
S_j
=
\max_{l\in F}
\left\|
e^{-i\beta_j\HM}\ket{l}
\right\|_1.
\]
This quantity measures the largest $\ell_1$-norm obtained when the mixer acts on a computational-basis state. A value close to one indicates that the resulting amplitude distribution has limited spreading, whereas larger values permit amplitudes to be distributed more broadly across configurations. Since the mixer unitary preserves the $\ell_2$-norm, the relation between the $\ell_1$- and $\ell_2$-norms gives
\[
1\leq S_j\leq\sqrt{|F|}.
\]

To characterize the cumulative effect of multiple mixer layers, we define
\[
\Lambda=\prod_{j=1}^{p}S_j^2.
\]
As we establish below, this quantity determines a uniform upper bound on the probability of measuring any individual configuration. It depends on the mixer unitaries but is independent of the cost angles, whose corresponding unitaries contribute only phases in the computational basis. 

We first express the QAOA objective in terms of contributions from sequences of intermediate configurations. This representation makes the role of the mixer amplitudes and cost-dependent phases explicit. It also provides a convenient starting point for bounding the output probabilities without imposing restrictions on the circuit parameters.

\begin{lemma}[Path expansion of the QAOA objective]
\label{lem:generalized}
For a depth-$p$ QAOA circuit, let
\[
q(x)=|\braket{x|\psi_p}|^2
\]
denote the probability of measuring a configuration $x\in F$. Then
\begin{equation}
\bra{\psi_p}\HC\ket{\psi_p}
=
\sum_{x\in F}C(x)q(x),
\label{eq:path-objective}
\end{equation}
where
\begin{equation}
\begin{aligned}
q(x)&=\frac{|A_x|^2}{|F|},\\
A_x
&=
\sum_{a_0,\ldots,a_{p-1}\in F}
\prod_{j=1}^{p}
K_j(a_j,a_{j-1})
e^{-i\gamma_jC(a_{j-1})},
\end{aligned}
\label{eq:path-sum}
\end{equation}
with $a_p=x$.
\end{lemma}

\rededit{The proof is provided in Supplementary Material, Section~\ref{S-supp:spreading}}

Each term in $A_x$ corresponds to a sequence of configurations connected through the mixer transitions. The magnitude of a term is determined by the mixer matrix elements, while the cost unitaries contribute phases that influence interference between different sequences. By discarding this interference through the triangle inequality, we obtain an upper bound that depends only on the spreading of the mixer layers. This leads to the following result.

\begin{lemma}[Upper bound on output probability]
\label{lem:pointwise}
For every $x\in F$, the probability of measuring $x$ at the end of a depth-$p$ QAOA circuit satisfies
\[
q(x)\leq\frac{\Lambda}{|F|}.
\]
\end{lemma}

\rededit{The proof is provided in Supplementary Material, Section~\ref{S-supp:spreading}}


\subsection{An Upper Bound on the QAOA Objective}

Lemma~\ref{lem:pointwise} establishes that the probability
of measuring any configuration is at most
$\Lambda/|F|$. We now use this restriction to derive
an upper bound on the expected QAOA objective. Since
the objective is a probability-weighted average of
the values $C(x)$, its maximum under this constraint
is obtained by assigning as much probability as
possible to the highest-value configurations, while
ensuring that the probabilities sum to one. We
formalize this as a linear program over all
probability distributions satisfying the pointwise
bound. Since these distributions include those
produced by QAOA, the optimum of the linear program
provides an upper bound on the achievable QAOA
objective.

\begin{theorem}[Objective bound from mixer spreading]
\label{thm:bathtub}
Let $N=|F|$, and arrange the objective values in decreasing order:
\[
\Cmax=C_{(1)}
\geq C_{(2)}
\geq\cdots
\geq C_{(N)}.
\]
Define
\[
\kappa=\frac{N}{\Lambda},
\qquad
m=\lfloor\kappa\rfloor.
\]
Then the expected QAOA objective satisfies
\[
\bra{\psi_p}\HC\ket{\psi_p}
\leq B(\Lambda),
\]
where
\begin{equation}
B(\Lambda)
=
\begin{cases}
\Cavg,
& \Lambda=1,
\\[2mm]
\dfrac{\Lambda}{N}
\left(
\displaystyle\sum_{i=1}^{m}C_{(i)}
+
(\kappa-m)C_{(m+1)}
\right),
& 1<\Lambda<N,
\\[3mm]
\Cmax,
& \Lambda\geq N.
\end{cases}
\label{eq:spreading-objective-bound}
\end{equation}

Here, $\Cavg$ denotes the average objective value
over all feasible configurations. The bound
satisfies
\[
B(1)=\Cavg,
\qquad
B(\Lambda)\leq\Cmax,
\]
and is non-decreasing in $\Lambda$.
Furthermore, if $C(x)\geq0$ for every $x\in F$,
then
\[
B(\Lambda)\leq\Lambda\,\Cavg.
\]
\end{theorem}

\rededit{The proof is provided in Supplementary Material, Section~\ref{S-supp:spreading}}

Theorem~\ref{thm:bathtub} relates the spreading of the mixer to the largest expected objective permitted by the resulting probability constraint. For $1\leq\Lambda\leq N$, the value $B(\Lambda)$ can be interpreted as the average objective over the best $1/\Lambda$ fraction of configurations, with a fractional contribution at the boundary when necessary. In particular, $B(1)$ equals the average objective, while increasing $\Lambda$ permits greater concentration on high-value configurations. Once $\Lambda\geq N$, the bound reaches $\Cmax$ and becomes uninformative beyond the maximum possible objective.

\begin{redchange}
The bound maximizes over all distributions satisfying the pointwise cap, not only those produced by QAOA. It therefore does not guarantee that a circuit can attain $B(\Lambda)$; interference and the restricted circuit dynamics may prevent this.
\end{redchange}

For nonnegative objective functions, the theorem yields a simpler layer-wise bound that does not require sorting the objective values. Although generally weaker than $B(\Lambda)$, it depends only on the average objective and the spreading of the individual mixer layers. This form can be useful when the full distribution of objective values is difficult to characterize.

\begin{corollary}[Layer-wise spreading bound]
\label{upper-bound}
If $C(x)\geq0$ for every $x\in F$, then
\[
\bra{\psi_p}\HC\ket{\psi_p}
\leq
\Cavg\prod_{j=1}^{p}S_j^2.
\]
\end{corollary}

\begin{proof}
Theorem~\ref{thm:bathtub} gives
\[
\bra{\psi_p}\HC\ket{\psi_p}
\leq B(\Lambda)
\leq\Lambda\,\Cavg.
\]
Substituting $\Lambda=\prod_{j=1}^{p}S_j^2$ gives the result.
\end{proof}

\subsection{Necessary Spreading for a Target Approximation Ratio}
\label{sec:necessary-spreading}

We now use the objective bound to obtain necessary conditions for achieving a target approximation ratio. Rather than determining the largest objective permitted by a given amount of spreading, we ask how much spreading is required to support a specified expected objective. This provides a way to rule out mixer choices or circuit depths that cannot reach a desired performance target, independently of how the cost angles are chosen.

Suppose that the target expected objective is at least $r\,\Cmax$, where $0<r\leq1$ and $\Cmax>0$. By Theorem~\ref{thm:bathtub}, achieving this target requires
\[
B(\Lambda)\geq r\,\Cmax.
\]
Since $B(\Lambda)$ is non-decreasing, this condition identifies a minimum spreading threshold. We define
\[
\Lambda_*(r)
=
\inf\left\{
\Lambda\geq1:
B(\Lambda)\geq r\,\Cmax
\right\}.
\]

The threshold $\Lambda_*(r)$ depends on the distribution of objective values over the feasible space, rather than on the details of any particular mixer. It therefore provides a problem-dependent necessary condition that every mixer must satisfy to achieve the target ratio. If $r\,\Cmax\leq\Cavg$, the target is already permitted by the uniform distribution, and the condition gives $\Lambda_*(r)=1$. For more demanding targets, the threshold quantifies how much additional spreading is necessary under our probability bound.

We can also translate this spreading requirement into a lower bound on circuit depth. If every mixer layer has bounded spreading, then the cumulative quantity $\Lambda$ cannot grow arbitrarily quickly with the number of layers. Combining this observation with the threshold $\Lambda_*(r)$ gives the following corollary.

\begin{corollary}[Necessary spreading and circuit depth]
\label{cor:necessary}
Suppose a depth-$p$ QAOA circuit achieves
\[
\bra{\psi_p}\HC\ket{\psi_p}
\geq r\,\Cmax.
\]
Then its mixer layers must satisfy
\[
\prod_{j=1}^{p}S_j^2
\geq
\Lambda_*(r).
\]

Furthermore, if $S_j\leq S_{\max}$ for every layer, where $S_{\max}>1$, then
\[
p
\geq
\frac{\ln\Lambda_*(r)}
{2\ln S_{\max}}.
\]
\end{corollary}

\begin{proof}
If the circuit achieves the target ratio, Theorem~\ref{thm:bathtub} implies
\[
r\,\Cmax
\leq
\bra{\psi_p}\HC\ket{\psi_p}
\leq B(\Lambda).
\]
By the definition of $\Lambda_*(r)$, this requires
\[
\Lambda
=
\prod_{j=1}^{p}S_j^2
\geq\Lambda_*(r).
\]

If $S_j\leq S_{\max}$ for every layer, then
\[
\Lambda
=
\prod_{j=1}^{p}S_j^2
\leq S_{\max}^{2p}.
\]
Combining these inequalities gives
\[
S_{\max}^{2p}\geq\Lambda_*(r).
\]
Taking logarithms, with $S_{\max}>1$, yields
\[
p\geq
\frac{\ln\Lambda_*(r)}
{2\ln S_{\max}},
\]
as required.
\end{proof}

\begin{redchange}
The threshold $\Lambda_*(r)$ identifies mixer parameters and depths that cannot reach a given target. We now compute the spreading quantities for two concrete mixer graphs to illustrate how the resulting restrictions depend on the mixer.
\end{redchange}

\section{Spreading Bounds for Specific Mixer Graphs}
\label{sec:vt}

We now apply the spreading bounds derived in the previous section to two commonly used mixers: the hypercube mixer and the complete-graph (Grover) mixer. These examples help illustrate how different mixer structures affect the spreading of quantum amplitude and the resulting bounds on QAOA performance.

Both mixers belong to the class of vertex-transitive graphs, where every vertex has the same structural properties. For such graphs, all columns of the mixer unitary have the same $\ell_1$-norm. Therefore, the spreading of each layer can be calculated using any computational-basis state $\ket{v_0}$:
\[
S_j(\beta_j)
=
\left\|e^{-i\beta_j\HM}\ket{v_0}\right\|_1.
\]
The general bound in Theorem~\ref{thm:bathtub} does not require vertex transitivity, but this property makes the spreading easier to evaluate.

\subsection{Hypercube Mixer}
\label{sec:case1_hypercube}

We first consider the standard transverse-field mixer,
\[
\HM=\sum_{k=1}^{n}X_k,
\]
acting on the full configuration space $F=\{0,1\}^n$. Its mixer graph is the $n$-dimensional hypercube, where two configurations are connected if they differ in exactly one bit.

Since all $X_k$ operators commute, the mixer unitary can be written as a product of single-qubit rotations. For one qubit,
\[
e^{-i\beta X}\ket{0}
=
\cos\beta\ket{0}
-i\sin\beta\ket{1}.
\]
The $\ell_1$-norm of this state is $|\cos\beta|+|\sin\beta|$. For $n$ qubits, the corresponding norm is the product of the individual norms, giving
\[
S_j(\beta_j)
=
\bigl(|\cos\beta_j|+|\sin\beta_j|\bigr)^n.
\]
Consequently,
\[
S_j^2
=
\bigl(1+|\sin 2\beta_j|\bigr)^n.
\]

At $\beta_j=\pi/4$, the spreading reaches its maximum value,
\[
S_j^2=2^n=|F|.
\]
If every layer uses this angle, then $\Lambda=|F|^p$. Even a single such layer makes $\Lambda\geq|F|$, so Theorem~\ref{thm:bathtub} gives only the trivial bound $\bra{\psi_p}\HC\ket{\psi_p}\leq\Cmax$.

For the spreading bound to impose a nontrivial
pointwise probability constraint, we require
$\Lambda<|F|$. For the hypercube mixer,
$|F|=2^n$, and therefore
\[
\Lambda
=
\left[
\prod_{j=1}^{p}
\bigl(1+|\sin 2\beta_j|\bigr)
\right]^n.
\]
Taking the $n$th root of $\Lambda<2^n$
gives
\begin{equation}
\Lambda<|F|
\iff
\prod_{j=1}^{p}
\bigl(1+|\sin 2\beta_j|\bigr)<2.
\label{eq:hypercube-window}
\end{equation}

Interestingly, this condition is independent
of the number of qubits. If all layers use
the same mixing angle $\beta$, it becomes
\[
|\sin 2\beta|<2^{1/p}-1
\approx\frac{\ln 2}{p},
\]
where the approximation holds for large $p$.

Thus, the spreading bound is useful only for certain ranges of mixing angles. In particular, it becomes uninformative near angles where $|\sin 2\beta|$ is large, including $\beta=\pi/4$. This illustrates a limitation of the $\ell_1$ spreading approach: a mixer can spread amplitude widely, but the resulting bound may then be too loose to constrain QAOA performance.

\subsection{Complete-Graph (Grover) Mixer}
\label{sec:case2_complete_graph}

We next consider the complete-graph mixer, which connects every configuration to every other configuration. Let $N=|F|$ and define the uniform state
\[
\ket{s}
=
\frac{1}{\sqrt{N}}\sum_{x\in F}\ket{x}.
\]

For a complete graph on $N$ configurations,
the adjacency matrix is $A_{K_N}=N\ket{s}\bra{s}-I$,
where
$\ket{s}=N^{-1/2}\sum_{x\in F}\ket{x}$
is the uniform superposition. Consequently,
$e^{-i\beta A_{K_N}}
=
e^{i\beta}
e^{-i\beta N\ket{s}\bra{s}}.$
Since the factor $e^{i\beta}$ is a global
phase and does not affect measurement
probabilities or expectation values, we
can equivalently use
$\HM=N\ket{s}\bra{s}$
as the complete-graph mixer Hamiltonian.

For a computational-basis state $\ket{l}$, the mixer acts as
\[
e^{-i\beta_j\HM}\ket{l}
=
\ket{l}
+
\frac{z_j}{N}\sum_{x\in F}\ket{x},
\]
where $z_j=e^{-i\beta_jN}-1.$
Taking the $\ell_1$-norm gives
\[
\begin{aligned}
S_j
&=
\left|1+\frac{z_j}{N}\right|
+
(N-1)\frac{|z_j|}{N}
\\
&\leq 1+|z_j|
\\
&=
1+2\left|\sin\frac{\beta_jN}{2}\right|
\leq3.
\end{aligned}
\]

Therefore, for a depth-$p$ circuit,
\[
\Lambda
=
\prod_{j=1}^{p}S_j^2
\leq9^p.
\]


\begin{redchange}
Unlike the hypercube mixer, the complete graph has layer-wise $\ell_1$ spreading bounded independently of $N$. Together these examples show how strongly the probability bound depends on mixer structure.
\end{redchange}

We next examine a different limitation of shallow QAOA circuits: how far information about the problem can propagate through their interactions. We study this through the causal reach of local objective terms and investigate how adding two-qubit interactions to the mixer can expand that reach.

\section{Expanding QAOA's Causal Reach through Mixer Augmentation}
\label{sec:augmentations}

\begin{redchange}
We now examine a different structural limitation: which parts of a MaxCut instance can influence a local cost term at shallow circuit depth. Using the MaxCut Hamiltonian and standard mixer defined in \Cref{sec:prelims-maxcut}, write $\HC=\sum_{(u,v)\in E}C_{uv}$, with $C_{uv}=(I-Z_uZ_v)/2$. We investigate how adding graph-aligned two-qubit mixer interactions changes the region relevant to each edge expectation.
\end{redchange}

Since the cost Hamiltonian is a sum of edge-local terms, the expected MaxCut objective can be written as
\[
\bra{\psi_p}\HC\ket{\psi_p}
=
\sum_{(u,v)\in E}\langle C_{uv}\rangle_p,
\]
where
\[
\langle C_{uv}\rangle_p
=
\bra{\psi_p}C_{uv}\ket{\psi_p}.
\]
This allows us to analyze the contribution of each edge separately. In particular, we can ask which qubits and interactions can influence the expectation value of a given edge term. Understanding how this set changes with the mixer is the main objective of this section.

\subsection{Augmenting the Mixer with Two-Qubit Interactions}

Our approach is to augment the standard transverse-field mixer with two-qubit interactions between vertices that are close to one another in the problem graph. Unlike the standard mixer, which contains only single-qubit rotations, the augmented mixers introduce interactions that can directly couple pairs of qubits. We consider two types of additional interactions, $XX$ and $YY$, and study how they affect the circuit's causal reach. We first define the construction for a general interaction distance $k$ before focusing on interactions along the edges of the problem graph.

\begin{definition}[$k$-hop mixer augmentation]
\label{def:augmentation}

Let $G=(V,E)$ be the problem graph, and let $k\geq1$ be an integer. We define the augmented edge set as
\[
\Eaug{k}
=
\left\{
\{i,j\}\subseteq V:
i\neq j,\;
\dist_G(i,j)\leq k
\right\}.
\]
Thus, $\Eaug{k}$ contains all pairs of distinct vertices whose distance in the problem graph is at most $k$.

For $P\in\{X,Y\}$, we construct an augmented mixer layer using single-qubit $X$ rotations and two-qubit $P_iP_j$ rotations on the pairs in $\Eaug{k}$. The corresponding unitary is
\[
U_{M,\ell}^{(k,P)}
=
\left(
\prod_{\{i,j\}\in\Eaug{k}}
e^{-i\beta'_\ell P_iP_j}
\right)
\left(
\prod_{v\in V}
e^{-i\beta_\ell X_v}
\right),
\]
where $\beta_\ell$ and $\beta'_\ell$ are the parameters controlling the single-qubit and two-qubit rotations, respectively. We refer to the resulting constructions as $X{+}XX$ when $P=X$ and $X{+}YY$ when $P=Y$.
\end{definition}

An important difference between these constructions is whether their generators commute. For the $X{+}XX$ mixer, all the generators commute, so the unitary can equivalently be written as
\[
U_{M,\ell}^{(k,X)}
=
\exp\left[
-i\left(
\beta_\ell\sum_{v\in V}X_v
+
\beta'_\ell
\sum_{\{i,j\}\in\Eaug{k}}X_iX_j
\right)
\right].
\]
For the $X{+}YY$ mixer, the single-qubit $X$ generators do not generally commute with the $YY$ interactions, so this simplification is not valid. We therefore use the ordered product of unitaries given in Definition~\ref{def:augmentation}, with the circuit implementations discussed in \cref{sec:mixer-circuits}.

In this work, we focus on the case $k=1$, for which
\[
\Eaug{1}=E.
\]
The additional two-qubit interactions then act on exactly the edges of the problem graph, so each $ZZ$ cost interaction is paired with an $XX$ or $YY$ mixer interaction on the same pair of qubits. This is the simplest augmentation that introduces two-qubit mixer interactions along the problem graph. We retain the general parameter $k$ to understand how the interaction distance affects causal reach, although larger values introduce practical difficulties discussed in \cref{sec:k-caveat}.

Additional two-qubit interactions have been considered in earlier mixer constructions, including the operator pool of ADAPT-QAOA~\cite{zhu2022adaptive} and the extended couplings studied in~\cite{kumar2026dealing}. Our contribution is not the introduction of these operators, but the analysis of their effect on the causal reach of shallow QAOA circuits. By aligning the additional interactions with the problem graph, we can characterize the resulting lightcones and evaluate their contributions for certain graph families. This also allows us to compare the $X{+}XX$ and $X{+}YY$ mixers under the same interaction structure.

\subsection{Lightcone Growth}
\label{sec:lightcone}

To understand the effect of mixer augmentation, we study how the support of an edge-local cost operator changes through the QAOA circuit. Consider the cost term $C_{uv}$ associated with an edge $(u,v)$. At finite circuit depth, its expectation value need not depend on the entire problem graph, but only on qubits that can influence the operator through the interactions implemented in the circuit. This set of qubits is called the \emph{lightcone} of $C_{uv}$, and its extent is determined by the circuit depth and the interactions available in each layer.

We first recall the lightcone bound for standard QAOA, where the mixer consists only of single-qubit $X$ rotations. In this case, the mixer cannot directly extend the support of an operator to another qubit, while the cost layer can introduce interactions along the edges of the problem graph. Consequently, the causal reach grows by at most one graph hop per QAOA layer. The following lemma states this standard result in our notation.

\begin{lemma}[Lightcone of standard QAOA]
\label{lem:lightcone}

Consider a depth-$p$ QAOA circuit with the standard transverse-field mixer $\HMk{0}$. For any edge $(u,v)\in E$, the Heisenberg-evolved operator $U_p^\dagger C_{uv}U_p$ is supported only on qubits whose graph distance from $u$ or $v$ is at most $p$. More precisely, its support is contained in
\[
L_p(u,v)
=
\left\{
w\in V:
\dist_G(w,u)\leq p
\;\text{or}\;
\dist_G(w,v)\leq p
\right\},
\]
where $U_p$ denotes the depth-$p$ QAOA unitary. Hence, $\langle C_{uv}\rangle_p$ can be evaluated using the interactions contained in this local neighborhood.
\end{lemma}

\begin{proof}
We follow the standard lightcone argument of Farhi et al.~\cite{farhi2014quantum}. For a single QAOA layer, the Heisenberg-evolved observable is
\[
U_1^\dagger C_{uv}U_1
=
e^{i\gamma_1\HC}
e^{i\beta_1\HMk{0}}
C_{uv}
e^{-i\beta_1\HMk{0}}
e^{-i\gamma_1\HC}.
\]
The transverse-field mixer consists of single-qubit rotations, and rotations acting on qubits other than $u$ and $v$ commute with $C_{uv}$. These rotations therefore cancel in the conjugation, so the mixer does not enlarge the support of the observable.

The cost unitary consists of commuting $ZZ$ interactions along the edges of $G$. Only interactions that overlap with the current support can affect the observable, and these interactions can introduce qubits adjacent to $u$ or $v$. Thus, after one layer, the evolved operator is supported within graph distance one of $\{u,v\}$. The same reasoning extends to additional layers: single-qubit mixer rotations preserve the current support, while each cost layer can enlarge it by at most one graph hop. After $p$ layers, the support is therefore contained in $L_p(u,v)$, proving the result.
\end{proof}

\begin{redchange}
The standard mixer can enlarge the lightcone by at most one hop per layer. Adding two-qubit mixer interactions allows a larger growth rate, as quantified next.
\end{redchange}

\begin{lemma}[Lightcone growth under $k$-hop augmentation]
\label{lem:lightcone-growth}

Consider a depth-$p$ QAOA circuit with the MaxCut cost Hamiltonian $\HC$ and either of the $k$-hop augmented mixer circuits from Definition~\ref{def:augmentation}. For any edge $(u,v)\in E$, the Heisenberg-evolved cost operator is supported within graph distance at most
\[
r(p)=(k+1)p
\]
of $\{u,v\}$. For the standard transverse-field mixer, which has no two-qubit augmentation, the corresponding radius bound is $r(p)=p$.
\end{lemma}

\rededit{The proof is provided in Supplementary Material, Section~\ref{S-supp:lightcone}}

\begin{redchange}
For the $k=1$ augmentations evaluated numerically, this gives $r(p)\leq2p$, compared with $r(p)\leq p$ for standard QAOA. Both $X{+}XX$ and $X{+}YY$ share this radius bound, allowing us to compare how different operators act within the same available interaction neighborhood.
\end{redchange}

\subsection{Reduction to Representative Lightcones}

The lightcone bounds are also useful for evaluating QAOA circuits on large graphs. Since each edge contribution depends only on a limited neighborhood, we do not always need to simulate the entire graph to calculate the objective. If two edges have the same local circuit structure and the circuit parameters are shared across the graph, their contributions to the expected objective are equal. We can therefore group edges according to their lightcone structure and evaluate only one representative from each group.

\begin{corollary}[Reduction to representative lightcone classes]
\label{cor:isoclasses}

Consider either the standard mixer or one of the graph-defined augmented mixer circuits from \cref{sec:mixer-circuits}, with parameters shared uniformly across vertices and augmentation edges. For fixed circuit parameters, the expectation value $\langle C_{uv}\rangle_p$ depends only on the isomorphism type of the edge-marked circuit lightcone around $(u,v)$, including the local cost and mixer interactions. Consequently, edges with isomorphic marked lightcones contribute equally to the expected objective.

Let $\mathcal{I}_p(G)$ be the set of distinct edge-marked lightcone classes, and let $m_c$ denote the number of edges belonging to class $c$. Then the expected objective can be written as
\[
\begin{aligned}
F(\boldsymbol\theta)
&=
\bra{\psi_p}\HC\ket{\psi_p}
\\
&=
\sum_{(u,v)\in E}\langle C_{uv}\rangle_p
\\
&=
\sum_{c\in\mathcal{I}_p(G)}
m_c\,\langle C\rangle_p^{(c)}(\boldsymbol\theta),
\end{aligned}
\]
where $\langle C\rangle_p^{(c)}$ is evaluated on a representative lightcone of class $c$.

For a graph with maximum degree $d\geq2$, a representative neighborhood of radius $r(p)$ contains at most $O(d^{r(p)})$ vertices. The radius satisfies $r(p)\leq p$ for the standard mixer and $r(p)\leq(k+1)p$ for the augmented mixers. Thus, at fixed depth and bounded degree, the size of each representative neighborhood is bounded independently of the total graph size.
\end{corollary}

\begin{redchange}
For graph families with a small number of edge-lightcone classes, the objective can be evaluated from a correspondingly small collection of local calculations, independent of total graph size at fixed depth and bounded degree. We apply this reduction below to cycles, circular ladders, and high-girth regular graphs; related locality methods have been used for standard QAOA~\cite{farhi2014quantum,wurtz2021maxcut}.
\end{redchange}

\section{Exact Local Evaluation via Representative Lightcones for $k$=1}
\label{sec:transitivity}

\subsection{Setup}

\begin{redchange}

We evaluate QAOA on cycles, circular ladders,
and high-girth regular bipartite graphs using
the representative-lightcone reduction of
\Cref{cor:isoclasses}. For these graph families,
edges have only one or two distinct local
neighborhoods at the circuit depths considered.
We therefore calculate the expectation value
for one representative edge of each type and
combine the results according to the number
of edges of that type.

This approach gives the exact expected objective
for any fixed choice of circuit parameters,
provided the graph is sufficiently large for
the assumed lightcone structure to hold.
However, the optimization of the circuit
parameters is performed numerically, so the
reported approximation ratios are the best
values found by the optimizer.
\end{redchange}

For each graph family, we define the approximation ratio at depth $p$ as
\[
\mathrm{AR}_p(G)
=
\frac{
\displaystyle
\max_{\boldsymbol{\theta}}
\bra{\psi_p(\boldsymbol{\theta})}
H_C
\ket{\psi_p(\boldsymbol{\theta})}
}{
\mathrm{OPT}(G)
},
\]
where
\[
\mathrm{OPT}(G)
=
\max_{z\in\{0,1\}^n}
\sum_{(u,v)\in E}C_{uv}(z).
\]
Here, $\boldsymbol{\theta}$ denotes the variational parameters of the mixer ansatz under consideration. The normalization is determined by the maximum cut value of the graph, while the numerator is evaluated using representative lightcones.

\paragraph{Cycle.}

Consider a cycle $C_n$ with $n$ vertices, and let $r(p)=(k+1)p$ denote the lightcone-radius bound. When $n\geq 2r(p)+2$, the neighborhood of an edge does not wrap around the cycle, and every edge has the same path-like representative lightcone. Therefore, all $n$ edges contribute the same local expectation, denoted by
$\langle C\rangle_p^{(\mathrm{path})}$. Since the circuit parameters are shared across all edges, the global objective satisfies
\[
\max_{\boldsymbol{\theta}}
\bra{\psi_p}H_C\ket{\psi_p}
=
n\max_{\boldsymbol{\theta}}
\langle C\rangle_p^{(\mathrm{path})}.
\]

For even $n$, the cycle is bipartite and
$\mathrm{OPT}(C_n)=n$, giving
\[
\mathrm{AR}_p(C_n)
=
\max_{\boldsymbol{\theta}}
\langle C\rangle_p^{(\mathrm{path})}.
\]
For odd $n$, the maximum cut contains $n-1$ edges, so the same local calculation gives
\[
\mathrm{AR}_p(C_n)
=
\frac{n}{n-1}
\max_{\boldsymbol{\theta}}
\langle C\rangle_p^{(\mathrm{path})}.
\]
The finite-size factor $n/(n-1)$ approaches one as $n$ increases, while the local edge expectation remains unchanged once the non-wrapping condition is satisfied.

\paragraph{Circular ladder.}

A circular ladder with $L$ columns consists of two cycles connected by $L$ rung edges. It has $2L$ vertices and $3L$ edges, with two edge orbits: the $2L$ rail edges and the $L$ rung edges. Since edges within each orbit have isomorphic lightcones, only two representative expectations are needed, denoted by
$\langle C\rangle_p^{(\mathrm{rail})}$ and
$\langle C\rangle_p^{(\mathrm{rung})}$. The global objective is therefore
\[
\begin{aligned}
\max_{\boldsymbol{\theta}}
\bra{\psi_p}H_C\ket{\psi_p}
=
\max_{\boldsymbol{\theta}}
\Bigl(
&2L\langle C\rangle_p^{(\mathrm{rail})}
\\
&+
L\langle C\rangle_p^{(\mathrm{rung})}
\Bigr).
\end{aligned}
\]

For even $L$, the ladder is bipartite and every edge can be included in the maximum cut, giving $\mathrm{OPT}=3L$. Hence,
\[
\mathrm{AR}_p(CL)
=
\max_{\boldsymbol{\theta}}
\frac{1}{3}
\left(
2\langle C\rangle_p^{(\mathrm{rail})}
+
\langle C\rangle_p^{(\mathrm{rung})}
\right).
\]
For odd $L$, the maximum cut has $3L-2$ edges, and the corresponding normalization is obtained by multiplying the expression above by $3L/(3L-2)$.

The lightcone must also remain local in the cyclic direction of the ladder. A sufficient condition for the neighborhood of radius $r(p)=(k+1)p$ not to wrap around the ladder is
\[
L\geq 2(k+1)p+2.
\]
Under this condition, the representative rail and rung lightcones are independent of the total number of columns. Their weighted expectations therefore determine the objective for every sufficiently large ladder, with only the appropriate normalization factor depending on the parity of $L$.

\paragraph{High-girth bipartite $(d+1)$-regular graphs: $BRG_d$.}

High-girth regular graphs provide another setting in which local QAOA expectations can be evaluated using representative neighborhoods. Basso et al.~\cite{basso2022highdepth} exploit the locally tree-like structure of such graphs in their analysis of QAOA. In our setting, a sufficient condition for the induced edge neighborhood of radius $r(p)=(k+1)p$ to remain tree-like is
\[
\operatorname{girth}(G)>2(k+1)p+2.
\]
This excludes additional cycles within the relevant neighborhood and ensures that its structure agrees with the corresponding regular tree.

For a bipartite $(d+1)$-regular graph satisfying this condition, every edge has the same representative tree-like lightcone. Since the graph is bipartite, $\mathrm{OPT}(G)=|E|$, and the approximation ratio reduces to the expectation of a single representative edge:
\[
\mathrm{AR}_p(BRG_d)
=
\frac{|E|\langle C\rangle_p^{BRG_d}}{|E|}
=
\langle C\rangle_p^{BRG_d}.
\]
Thus, for the depths considered here, the objective can be evaluated without simulating the entire regular graph. The computation depends only on the local degree and circuit depth, provided the graph has sufficiently large girth.

\subsection{Best-Found Performance at Fixed Depth for $k=1$}
\label{sec:results}

\subsubsection{Experimental Setup}
\label{sec:setup}

We evaluate the three graph families using the representative-lightcone reduction described above. For each edge orbit, we construct the lightcone of a representative edge using the breadth-first-search procedure underlying \cref{lem:lightcone}. At $k=1$, the lightcone-radius bound is $2p$, corresponding to the propagation allowed by the cost and augmented mixer interactions over $p$ layers. We then construct the cost and mixer operators on the relevant lightcone, initialize its qubits in the uniform superposition, and evaluate the expectation of the representative edge-cost operator.

The local calculations are performed using explicit state-vector simulation when the lightcones are sufficiently small. For larger lightcones, we use tensor-network contraction to evaluate the same local expectation without storing the full state vector. This is necessary primarily for the high-girth regular graph family, where the number of vertices in a representative neighborhood grows as $O(d^{2p})$, rather than linearly with $p$. In both cases, the lightcone reduction introduces no approximation to the fixed-parameter objective; numerical error arises only from the evaluation method and finite-precision computation.

Because the variational parameters are shared across the full graph, the expectations of different edge orbits must be optimized jointly. We therefore define the normalized local objective
\[
\mathcal O(\boldsymbol{\theta})
=
\sum_{\text{orbit }o}
w_o
\langle C\rangle_p^{(o)}(\boldsymbol{\theta}),
\qquad
\sum_o w_o=1,
\]
where $w_o$ is the fraction of edges belonging to orbit $o$. The cycle and high-girth regular graphs have a single orbit with weight one, while the circular ladder has weights
\[
w_{\mathrm{rail}}=\frac23,
\qquad
w_{\mathrm{rung}}=\frac13.
\]
We maximize this weighted objective over a common parameter vector, ensuring that the resulting values correspond to a single globally defined ansatz rather than independently optimized edge expectations.

The numerical optimization uses BFGS with multiple random initializations. Each reported value is the best result obtained across several independent optimization runs, each using ten random restarts, or six restarts for the more computationally expensive regular-tree evaluations. The corresponding parameter vectors and implementation details are available from the authors upon reasonable request. Since BFGS is a local optimizer, these values are best-found results and provide lower bounds on the globally optimized approximation ratios, rather than certificates of global optimality.

\subsubsection{Mixer Circuits}
\label{sec:mixer-circuits}

We compare the standard transverse-field mixer with two augmented mixers at $k=1$. At this augmentation level, the additional two-qubit interactions act on the original problem edges, so the augmented edge set is $E_1=E$. Each QAOA layer consists of a cost unitary followed by the selected mixer, with
\[
U_C(\gamma_\ell)
=
\prod_{(a,b)\in E}
e^{-i\gamma_\ell C_{ab}}.
\]
The three ansatz constructions differ in their mixer operators and in the number of variational parameters assigned to each layer.

\paragraph{Standard $X$ mixer.}

The standard transverse-field mixer is
\[
U_M^X(\beta_\ell)
=
\prod_{v\in V}e^{-i\beta_\ell X_v}.
\]
All single-qubit terms commute, so the ordering of these gates does not affect the resulting unitary. Each layer has two variational parameters, $(\gamma_\ell,\beta_\ell)$, corresponding to the cost and mixer rotations. We use this ansatz as the baseline for evaluating the benefit of edge-based mixer augmentation.

\paragraph{Augmented $X+XX$ and $X+YY$ mixers.}

For the $XX$ augmentation, the mixer is
\[
\begin{aligned}
U_M^{X+XX}(\beta_\ell,\beta'_\ell)
={}&
\left(
\prod_{(u,v)\in E}
e^{-i\beta'_\ell X_uX_v}
\right)
\\
&\times
\left(
\prod_{v\in V}
e^{-i\beta_\ell X_v}
\right).
\end{aligned}
\]
Similarly, the $YY$ augmentation is defined by
\[
\begin{aligned}
U_M^{X+YY}(\beta_\ell,\beta'_\ell)
={}&
\left(
\prod_{(u,v)\in E}
e^{-i\beta'_\ell Y_uY_v}
\right)
\\
&\times
\left(
\prod_{v\in V}
e^{-i\beta_\ell X_v}
\right).
\end{aligned}
\]

In both constructions, the additional edge interactions share a single variational angle $\beta'_\ell$ within each layer. Consequently, the augmented ansatz has three parameters per layer,
$(\gamma_\ell,\beta_\ell,\beta'_\ell)$, compared with two for the standard mixer. The edge interactions within either augmentation commute with one another, so their internal ordering is immaterial. However, the $YY$ augmentation does not generally commute with the transverse-field $X$ rotations, whereas the $XX$ augmentation does; this distinction will be relevant when interpreting their performance.

\paragraph{Exclusion of the combined $XX+YY$ mixer.}

A further possibility is to combine $XX$ and $YY$ interactions on each augmentation edge. Unlike the separate $XX$ and $YY$ constructions, combined interactions on overlapping edges generally do not commute, so a gate-level implementation requires a specified ordering or scheduling procedure. One possibility is to use an edge-coloring of the augmentation graph and apply interactions on vertex-disjoint edge sets in successive steps. However, the resulting local circuit can depend on the coloring and gate order, in addition to the underlying graph topology.

This additional dependence complicates the representative-lightcone reduction. Two edges with isomorphic uncolored neighborhoods may experience different ordered gate sequences, so their local expectations need not coincide solely because the graph neighborhoods are isomorphic. An exact isomorphism-class reduction for such an implementation would need to incorporate the relevant scheduling information into the equivalence classes. To retain the direct topology-based reduction used in this section, we restrict the main comparison to commuting edge-interaction families, namely $XX$ and $YY$.

\subsubsection{Results}
\label{sec:results-tables}

The best-found approximation ratios for the three graph families are reported in \cref{tab:cycle-results,tab:ladder-results,tab:tree-results}. In each case, we compare the standard $X$ mixer and the two augmented variants at equal QAOA depth. The tables also report the number of variational parameters, allowing the performance differences to be considered alongside the additional parameter cost of augmentation.

\begin{table}[htbp]
\centering
\begin{tabular}{@{}lccccc@{}}
\toprule
$p$ & $\mathrm{AR}_p^{X}$
& $\mathrm{AR}_p^{X+XX}$
& $\mathrm{AR}_p^{X+YY}$
& params ($X$) & params (aug.) \\
\midrule
1 & \texttt{0.7500} & \texttt{0.7829} & \texttt{0.8247} & 2 & 3 \\
2 & \texttt{0.8333} & \texttt{0.8656} & \texttt{0.8987} & 4 & 6 \\
3 & \texttt{0.8750} & \texttt{0.9034} & \texttt{0.9277} & 6 & 9 \\
4 & \texttt{0.9000} & \texttt{0.9198} & \texttt{0.9424} & 8 & 12 \\
\bottomrule
\end{tabular}
\caption{Best-found approximation ratios for the cycle $C_n$ with $k=1$, evaluated using representative lightcones.}
\label{tab:cycle-results}
\end{table}

\begin{table}[htbp]
\centering
\begin{tabular}{@{}lccccc@{}}
\toprule
$p$ & $\mathrm{AR}_p^{X}$
& $\mathrm{AR}_p^{X+XX}$
& $\mathrm{AR}_p^{X+YY}$
& params ($X$) & params (aug.) \\
\midrule
1 & \texttt{0.6924} & \texttt{0.7129} & \texttt{0.7806} & 2 & 3 \\
2 & \texttt{0.7918} & \texttt{0.8118} & \texttt{0.8525} & 4 & 6 \\
\bottomrule
\end{tabular}
\caption{Best-found approximation ratios for the circular ladder with $k=1$, using the jointly optimized rail and rung expectations.}
\label{tab:ladder-results}
\end{table}

\begin{table}[htbp]
\centering
\begin{tabular}{@{}lccccc@{}}
\toprule
$d$ & $\mathrm{AR}_1^{X}$
& $\mathrm{AR}_1^{X+XX}$
& $\mathrm{AR}_1^{X+YY}$
& params ($X$) & params (aug.) \\
\midrule
2 & \texttt{0.6924} & \texttt{0.7050} & \texttt{0.7455} & 2 & 3 \\
3 & \texttt{0.6623} & \texttt{0.6695} & \texttt{0.7063} & 2 & 3 \\
4 & \texttt{0.6431} & \texttt{0.6479} & \texttt{0.6815} & 2 & 3 \\
5 & \texttt{0.6294} & \texttt{0.6329} & \texttt{0.6639} & 2 & 3 \\
\bottomrule
\end{tabular}
\caption{Best-found approximation ratios for high-girth bipartite $(d+1)$-regular graphs at $p=1$ and $k=1$.}
\label{tab:tree-results}
\end{table}

\begin{redchange}
Across the three graph families, both augmented mixers improve the best-found approximation ratios over $X$ at equal depth, and $X+YY$ outperforms $X+XX$ in every tested setting. On the cycle, the $X+YY$ advantage decreases from about $7.5$ to $4.2$ percentage points as $p$ increases from one to four. On the circular ladder, the corresponding improvements are about $8.8$ and $6.1$ points at depths one and two, while the gains on high-girth regular graphs decrease with degree. The tables provide the complete numerical values.
\end{redchange}

The equal-depth comparison demonstrates the improvement associated with adding edge-based interactions, but the augmented mixers also have a larger lightcone-radius bound. At $k=1$, an augmented circuit of depth $p$ has radius at most $2p$, matching the radius bound of a standard $X$ circuit at depth $2p$. Comparing these circuits therefore helps distinguish improvements associated with increased causal reach from those arising from the structure of the mixer itself. Equal lightcone-radius bounds do not, however, imply identical operator propagation or equivalent variational families.

The results show a substantial difference between the two augmentations under this reach-matched comparison. On the cycle, $X+YY$ at depths one and two achieves ratios of $0.8247$ and $0.8987$, compared with $0.8333$ and $0.9000$ for the standard $X$ mixer at depths two and four. On the circular ladder, $X+YY$ at depth one achieves $0.7806$, compared with $0.7918$ for standard $X$ at depth two. Thus, in the available comparisons, the $X+YY$ mixer comes within approximately $0.1$--$1.1$ percentage points of the deeper standard mixer.

\begin{redchange}
In the same reach-matched comparisons, $X+XX$ remains further below the deeper standard-$X$ baseline, illustrating that the operators used to extend the lightcone affect the resulting performance. An augmented circuit at depth $p$ uses $3p$ variational parameters, compared with $4p$ for standard $X$ at depth $2p$. The constructions have different gate types and scheduling costs, so the numerical comparison does not establish a hardware-depth advantage. All reported ratios are best-found values, not certified ansatz optima.
\end{redchange}

\subsubsection{Parameter Budget}
\label{sec:param-budget}

\begin{redchange}
To probe whether the extra variational angle alone could explain the improvement, we considered an exploratory standard-$X$ control with two independent transverse-field angles per layer. The qubits were grouped separately within each representative lightcone, and the additional angle produced negligible improvement in the tests. Because this grouping is local rather than a globally consistent ansatz on the original graph, it is not a strict parameter-matched baseline; a global control is needed for a conclusive comparison.
\end{redchange}

\begin{redchange}
\subsubsection{Other mixer variants}
We also tested standalone two-qubit mixers and the $X+ZZ$ and $X+AA(\theta)$ augmentations. Their results and the symmetry arguments explaining several cases are presented in Supplementary Section~\ref{S-supp:additional-mixers}.
\end{redchange}

\subsection{Scope of the $k$-Hop Mixer Augmentation}
\label{sec:k-caveat}

\begin{redchange}
The lightcone analysis applies to general $k$-hop augmentation, whereas the numerical results are restricted to $k=1$. At this level the extra $XX$ or $YY$ interactions act on pairs already connected by the MaxCut cost Hamiltonian, so no additional logical edges are required. This does not remove physical routing or native-gate costs, which depend on the device and implementation.

For $k\geq2$, longer-range interactions can increase these costs. Shared angles can still keep the parameter count small, but additional interaction classes may demand more flexible parameterizations. Assessing whether the resulting increase in causal reach is worth those resources is left for future work.
\end{redchange}

\section{Discussion and Future Work}
\label{sec:discussion}

\begin{redchange}
Our two analyses examine complementary features of shallow-QAOA mixers. Transport concerns amplitude movement across feasible configurations, while causal reach concerns which graph interactions can influence a local observable. They constrain different aspects of the circuit and need not predict performance in the same way.
\end{redchange}

For amplitude spreading, our perturbative analysis
shows that the small-$\beta$ and small-$\gamma$
bounds are closely related. Although one is expressed
using cost differences between connected configurations
and the other uses the eigenvalues and eigenvectors of
the mixer, both involve the same graph Dirichlet energy.
This equality follows from an exact graph-Laplacian
identity for connected regular mixer graphs, even though
the two bounds are derived under different parameter
restrictions. It also shows why the spectral gap alone
cannot fully describe the mixer's effect on the
objective: the result depends on how the cost function
is distributed across the mixer eigenmodes. Our
truncation-error analysis further relates these
first-order bounds to the exact QAOA objective.

\begin{redchange}
Outside the perturbative regime, the $\ell_1$-spreading bound yields a probability cap and a target-dependent depth restriction. Its main limitation is the loss of relative phase information when computational paths are bounded by absolute values. Developing a phase-sensitive refinement could therefore sharpen the result.
\end{redchange}

\begin{redchange}
The lightcone analysis offers a parallel lesson. Although both $X+XX$ and $X+YY$ access the same radius bound and improve the best-found objective at equal depth, their observed performance differs. The effect of an augmentation therefore depends on the operators used within the available lightcone, not just the geometric extent of that lightcone. The distinction between commuting $XX$ interactions and the noncommuting $YY$ block is a useful starting point for explaining this behavior analytically.
\end{redchange}

There are several directions in which these results
could be extended. One is to improve the spreading
bounds by accounting for interference between
computational paths, which may lead to tighter bounds
on the achievable objective. Another is to extend
the perturbative analysis to nonregular mixer graphs,
where the uniform state need not be an eigenvector
of the mixer. It would also be useful to study the
low-depth behavior of $X+XX$ and $X+YY$
analytically, particularly to understand why
their different operator structures lead to
different numerical performance.

Finally, our numerical study is restricted to
highly structured graph families and $k=1$
augmentations. Future experiments could consider
more general problem graphs and higher-order
mixer interactions. These comparisons should also
account for differences in parameter count, gate
count, circuit depth, and hardware connectivity
to determine when the additional interactions
provide a practical benefit.

\section{Conclusion}
\label{sec:conclusion}

\begin{redchange}
We studied two structural properties of shallow-QAOA mixers: configuration-space transport and finite-depth causal reach. The analytical bounds connect perturbative objective behavior to the mixer graph and constrain probability concentration beyond the small-angle regime. The lightcone and numerical analyses show how graph-aligned two-qubit interactions affect MaxCut performance, with different operators yielding different outcomes even when their interaction radius is the same. Together, these results motivate mixer designs that consider both the graph structure and the quantum dynamics generated by the chosen operators.

\end{redchange}

\paragraph{Code Availability.}

The code used to implement the QAOA
simulations, perform parameter optimization, and produce
the numerical results reported in this work is available from the authors
upon reasonable request.

\paragraph{Declaration of AI-Assisted Writing.}

The scientific ideas, methodology, mathematical analysis, experimental design, and interpretation of the results
presented in this work were completely carried out by the authors. Generative AI tools
were used only to assist with parts of the manuscript writing process,
including improving language, grammar, clarity, and presentation. The authors
reviewed the final manuscript and take full responsibility for its content.

\bibliographystyle{plainurl}
\bibliography{references}

\clearpage
\begin{center}
{\Large\textbf{Supplementary Material for\\ Structural Bottlenecks in QAOA Mixer Design:\\ Transport and Causal Reach}}
\end{center}
\bigskip

\setcounter{section}{0}
\renewcommand{\thesection}{S\arabic{section}}
\renewcommand{\theHsection}{S\arabic{section}}
\setcounter{equation}{0}
\numberwithin{equation}{section}
\renewcommand{\theHequation}{\theequation}
\setlength{\parindent}{0pt}
\setlength{\parskip}{5pt plus 1pt}
\allowdisplaybreaks[1]

\section{Supporting Proofs for Configuration-Space Transport}
\label{S-supp:transport}
\subsection*{Notation and perturbative expectation values}

Let $F\subseteq\{0,1\}^n$ be a finite feasible set and let
$C:F\to\mathbb{R}$ be an objective function to be maximized.
The cost Hamiltonian and the uniform starting state are
\[
\HC=\sum_{x\in F}C(x)\ket{x}\bra{x},\qquad
\ket{\psi_0}=\frac{1}{\sqrt{|F|}}\sum_{x\in F}\ket{x}.
\]
The QAOA circuit has cost and mixer unitaries
$U_C(\gamma)=e^{-i\gamma\HC}$ and $U_M(\beta)=e^{-i\beta\HM}$,
with
\[
\ket{\psi_p}=U_M(\beta_p)U_C(\gamma_p)\cdots
U_M(\beta_1)U_C(\gamma_1)\ket{\psi_0}.
\]
In the transport analysis, $\HM$ is the real symmetric adjacency
matrix of an unweighted, undirected mixer graph on $F$.
Write $E_{\HM}$ for its undirected edge set, and set
\[
\bar C=\frac{1}{|F|}\sum_{x\in F}C(x),\qquad
C_{\max}=\max_x C(x),\qquad C_{\min}=\min_x C(x).
\]
Where needed, a constant may be subtracted from $C$ to make it
nonnegative; such a shift does not change the optimizing parameters.
For the perturbative bounds, the parameters satisfy
$0\leq\beta_j\leq\eta$ and $0\leq\gamma_j\leq\delta$.

Define the Heisenberg-recursion operators
\[
\mathcal{O}^{(p)}=\HC,\qquad
\mathcal{O}^{(\ell-1)}=
 e^{i\gamma_\ell\HC}e^{i\beta_\ell\HM}
 \mathcal{O}^{(\ell)}e^{-i\beta_\ell\HM}e^{-i\gamma_\ell\HC},
\]
so that $\bra{\psi_p}\HC\ket{\psi_p}
=\bra{\psi_0}\mathcal O^{(0)}\ket{\psi_0}$.
The nested commutators are defined by
\[
\com_A^n(B)=[A,[A,\ldots,[A,B]]],\qquad\com_A^0(B)=B,
\]
where the first expression contains $n$ commutators.
Put
\[
\Gamma_m=\sum_{j=1}^{m}\gamma_j,\qquad
\Sigma_m=\sum_{j=m}^{p}\beta_j.
\]
The first-order approximations to the evolved cost operator, in the
mixing and cost angles respectively, are
\begin{align}
A_p&=\HC+i\sum_{m=1}^{p}\beta_m
 e^{i\Gamma_m\HC}[\HM,\HC]e^{-i\Gamma_m\HC},\label{S-eq:first-order-operator}\\
B_p&=e^{i\Sigma_1\HM}\HC e^{-i\Sigma_1\HM}
+i\sum_{m=1}^{p}\gamma_m e^{i(\Sigma_1-\Sigma_m)\HM}
 [\HC,e^{i\Sigma_m\HM}\HC e^{-i\Sigma_m\HM}]
 e^{-i(\Sigma_1-\Sigma_m)\HM}.\label{S-eq:first-order-gamma}
\end{align}
Their expectations are
\begin{equation}
\begin{aligned}
\bra{\psi_0}A_p\ket{\psi_0}
={}&\bar C+\frac{2}{|F|}\sum_{m=1}^{p}\beta_m
\sum_{\substack{x,y\in F\\x<y}}
(\HM)_{xy}\bigl(C(y)-C(x)\bigr)
\sin\!\bigl(\Gamma_m(C(y)-C(x))\bigr),
\end{aligned}
\label{S-eq:first-order-expectation}
\end{equation}
where $x<y$ denotes any fixed ordering used to count each undirected
edge once. If the mixer graph is connected and $d$-regular, let
\[
\HM=\sum_b\lambda_b\ket{\phi_b}\bra{\phi_b},\quad
\lambda_0=d,\quad\ket{\phi_0}=\ket{\psi_0},\quad
C_{ab}=\bra{\phi_a}\HC\ket{\phi_b}.
\]
Then
\begin{equation}
\bra{\psi_0}B_p\ket{\psi_0}
=\bar C+2\sum_{m=1}^{p}\gamma_m
\sum_{b\ne0}|C_{0b}|^2\sin\bigl(\Sigma_m(\lambda_0-\lambda_b)\bigr).
\label{S-eq:gamma-expectation}
\end{equation}
We also use the Dirichlet energy and uniform cost variance,
\[
\mathcal E_M=\sum_{\{x,y\}\in E_{\HM}}(C(x)-C(y))^2,\qquad
\Var(C)=\frac{1}{|F|}\sum_{x\in F}(C(x)-\bar C)^2.
\]

\subsection{Proofs of the perturbative results}
\label{S-supp:perturbative}

\begin{restatedresult}{Lemma~\ref{lem:small-beta} --- Expansion in $\beta$}
For each layer $\ell$, the evolved cost operator admits the exact expansion
\begin{equation}
\mathcal{O}^{(\ell-1)}
=
e^{i\gamma_\ell\HC}
\left(
\sum_{n=0}^{\infty}
\frac{(i\beta_\ell)^n}{n!}
\com_{\HM}^{n}(\mathcal{O}^{(\ell)})
\right)
e^{-i\gamma_\ell\HC}.
\label{S-eq:small-beta-expansion}
\end{equation}
Applying this expansion successively for $\ell=p,\ldots,1$ gives an expression for $\mathcal{O}^{(0)}$ in terms of the mixing angles and nested commutators.
\end{restatedresult}

\begin{proof}[Proof of Lemma~\ref{lem:small-beta}]
By Hadamard's lemma, for any operators $A$ and $B$,
\[
e^A B e^{-A}
=
\sum_{n=0}^{\infty}
\frac{1}{n!}\com_A^n(B).
\]
Applying this identity with $A=i\beta_\ell\HM$ and
$B=\mathcal{O}^{(\ell)}$, we obtain
\[
e^{i\beta_\ell\HM}
\mathcal{O}^{(\ell)}
e^{-i\beta_\ell\HM}
=
\sum_{n=0}^{\infty}
\frac{(i\beta_\ell)^n}{n!}
\com_{\HM}^{n}(\mathcal{O}^{(\ell)}).
\]
Substituting this expansion into the recursive definition
of $\mathcal{O}^{(\ell-1)}$ gives
\[
\mathcal{O}^{(\ell-1)}
=
e^{i\gamma_\ell\HC}
\left(
\sum_{n=0}^{\infty}
\frac{(i\beta_\ell)^n}{n!}
\com_{\HM}^{n}(\mathcal{O}^{(\ell)})
\right)
e^{-i\gamma_\ell\HC},
\]
which proves the result.
\end{proof}

\begin{restatedresult}{Lemma~\ref{lem:small-gamma} --- Expansion in $\gamma$}
For each layer $\ell$, the evolved cost operator follows the exact expansion
\begin{equation}
\mathcal{O}^{(\ell-1)}
=
\sum_{n=0}^{\infty}
\frac{(i\gamma_\ell)^n}{n!}
\com_{\HC}^{n}
\left(
e^{i\beta_\ell\HM}
\mathcal{O}^{(\ell)}
e^{-i\beta_\ell\HM}
\right).
\label{S-eq:small-gamma-expansion}
\end{equation}
Applying this expansion successively for $\ell=p,\ldots,1$ gives an expression for $\mathcal{O}^{(0)}$ in terms of the cost angles and nested commutators.
\end{restatedresult}

\begin{proof}[Proof of Lemma~\ref{lem:small-gamma}]
This follows directly from Hadamard's lemma, now applied to conjugation by $e^{i\gamma_\ell\HC}$.
\end{proof}

The small-$\beta$ estimates use the parameter condition
\begin{equation}
0\leq\beta_i\leq\eta,\quad 0\leq\gamma_i\leq\delta,
\quad \delta\leq\frac{\pi}{2p(C_{\max}-C_{\min})},
\label{S-eq:delta-condition}
\end{equation}
when $C_{\max}>C_{\min}$. If $C$ is constant, the corresponding
Dirichlet energy vanishes and the stated estimates are immediate.

\begin{restatedresult}{Proposition~\ref{prop:linear-bounds} --- First-order small-$\beta$ bound}
Under the conditions of \eqref{S-eq:delta-condition},
\begin{equation}
\begin{aligned}
\bar C
+
\frac{2\eta\delta\,p(p+1)}{\pi|F|}
\mathcal{E}_M
&\leq
\max_{\substack{0\leq\beta_i\leq\eta\\
0\leq\gamma_i\leq\delta}}
\bra{\psi_0}A_p\ket{\psi_0}
\\
&\leq
\bar C
+
\frac{\eta\delta\,p(p+1)}{|F|}
\mathcal{E}_M,
\end{aligned}
\label{S-eq:linear-bounds}
\end{equation}
where
\[
\mathcal{E}_M
=
\sum_{\{x,y\}\in E_{\HM}}
\bigl(C(x)-C(y)\bigr)^2,
\]
and $E_{\HM}$ is the set of undirected edges of the mixer graph.
\end{restatedresult}

\begin{proof}[Proof of Proposition~\ref{prop:linear-bounds}]
From \eqref{S-eq:first-order-expectation}, 
and since $\HM$ is the adjacency matrix of an unweighted mixer graph, $(\HM)_{xy}=1$ for connected configurations and zero otherwise, the first-order objective is
\[
\begin{aligned}
\bra{\psi_0}A_p\ket{\psi_0}
=
\bar C+
\frac{2}{|F|}
\sum_{m=1}^{p}\beta_m
\sum_{\{x,y\}\in E_{\HM}}
\Delta_{xy}\sin(\Gamma_m\Delta_{xy}),
\end{aligned}
\]
where $\Delta_{xy}=C(y)-C(x)$.

Under the restriction in \eqref{S-eq:delta-condition},
\[
0\leq\Gamma_m|\Delta_{xy}|\leq\frac{\pi}{2}.
\]
Since
\[
\Delta_{xy}\sin(\Gamma_m\Delta_{xy})
=
|\Delta_{xy}|\sin(\Gamma_m|\Delta_{xy}|),
\]
each contribution is nonnegative and non-decreasing in the allowed parameters. The first-order objective is therefore maximized at
\[
\beta_m=\eta,\qquad \gamma_m=\delta,
\]
for all $m$. At these values,
\[
\Gamma_m=m\delta.
\]
Applying Jordan's inequality,
\[
\frac{2}{\pi}t\leq\sin t\leq t,
\qquad 0\leq t\leq\frac{\pi}{2},
\]
we obtain
\[
\frac{2}{\pi}\Gamma_m\Delta_{xy}^2
\leq
\Delta_{xy}\sin(\Gamma_m\Delta_{xy})
\leq
\Gamma_m\Delta_{xy}^2.
\]

Substituting these inequalities into the first-order objective gives
\[
\begin{aligned}
\bar C+
\frac{4\eta}{\pi|F|}
\left(\sum_{m=1}^{p}\Gamma_m\right)\mathcal E_M
&\leq
\max\bra{\psi_0}A_p\ket{\psi_0}\\
&\leq
\bar C+
\frac{2\eta}{|F|}
\left(\sum_{m=1}^{p}\Gamma_m\right)\mathcal E_M.
\end{aligned}
\]

Finally, since
\[
\sum_{m=1}^{p}\Gamma_m
=
\delta\sum_{m=1}^{p}m
=
\frac{\delta p(p+1)}{2},
\]
we obtain
\[
\begin{aligned}
\bar C+
\frac{2\eta\delta\,p(p+1)}{\pi|F|}
\mathcal E_M
&\leq
\max\bra{\psi_0}A_p\ket{\psi_0}\\
&\leq
\bar C+
\frac{\eta\delta\,p(p+1)}{|F|}
\mathcal E_M.
\end{aligned}
\]
This completes the proof.
\end{proof}

For the complementary small-$\gamma$ regime, the condition is
\begin{equation}
0\leq\gamma_i\leq\delta,\quad 0\leq\beta_i\leq\eta,
\quad \eta\leq\frac{\pi}{2p(\lambda_0-\lambda_{\min})},
\label{S-eq:eta-condition}
\end{equation}
where $\lambda_{\min}$ is the smallest eigenvalue of $\HM$.

\begin{restatedresult}{Proposition~\ref{prop:gamma-bound} --- First-order small-$\gamma$ bound}
Under the conditions of \eqref{S-eq:eta-condition},
\begin{equation}
\begin{aligned}
\bar C
+
\frac{2\delta\eta\,p(p+1)}{\pi}
\sum_{b\ne0}
(\lambda_0-\lambda_b)|C_{0b}|^2
&\leq
\max_{\substack{0\leq\beta_i\leq\eta\\
0\leq\gamma_i\leq\delta}}
\bra{\psi_0}B_p\ket{\psi_0}
\\
&\leq
\bar C
+
\delta\eta\,p(p+1)
\sum_{b\ne0}
(\lambda_0-\lambda_b)|C_{0b}|^2.
\end{aligned}
\label{S-eq:gamma-bound}
\end{equation}
\end{restatedresult}

\begin{proof}[Proof of Proposition~\ref{prop:gamma-bound}]
Using \eqref{S-eq:eta-condition}, we have
\[
0\leq
\Sigma_m(\lambda_0-\lambda_b)
\leq
p\eta(\lambda_0-\lambda_{\min})
\leq\frac{\pi}{2}.
\]
Since $\sin t$ is nonnegative and increasing in
this interval, the first-order objective is
maximized at $\beta_k=\eta$ and $\gamma_k=\delta$
for all $k$, giving $\Sigma_m=(p-m+1)\eta$.

Applying Jordan's inequality,
\[
\frac{2}{\pi}t\leq\sin t\leq t,
\qquad 0\leq t\leq\frac{\pi}{2},
\]
to the spectral expression for
$\bra{\psi_0}B_p\ket{\psi_0}$, and using
\[
\sum_{m=1}^{p}\Sigma_m
=
\eta\sum_{m=1}^{p}(p-m+1)
=
\frac{\eta p(p+1)}{2},
\]
gives the stated lower and upper bounds.
\end{proof}

\begin{restatedresult}{Proposition~\ref{prop:equivalence} --- Equivalence of the perturbative regimes}
Let $\HM$ be the adjacency matrix of a connected $d$-regular graph on $F$. Then
\begin{equation}
\mathcal{E}_M
=
|F|
\sum_{b\ne0}
(\lambda_0-\lambda_b)|C_{0b}|^2.
\label{S-eq:equivalence-identity}
\end{equation}
Consequently, the leading-order upper and lower expressions in \eqref{S-eq:linear-bounds} and \eqref{S-eq:gamma-bound} are identical after this substitution. The exact-objective bounds obtained after including truncation errors may have different remainder terms.
\end{restatedresult}

\begin{proof}[Proof of Proposition~\ref{prop:equivalence}]
Let
$L=dI-\HM$
be the graph Laplacian of the mixer graph, and let
$\vec{c}=(C(x))_{x\in F}$
be the vector of objective values.

Since $\HM$ is the adjacency matrix of an unweighted,
undirected graph, the standard graph-Laplacian identity
gives
\begin{equation}
\begin{aligned}
\mathcal E_M
&=
\sum_{\{x,y\}\in E_{\HM}}
(C(x)-C(y))^2
\\
&=
\frac12\sum_{x,y\in F}
(\HM)_{xy}(C(x)-C(y))^2
\\
&=
\vec{c}^{\top}L\vec{c}.
\end{aligned}
\label{S-eq:step1}
\end{equation}

Since the mixer graph is connected and $d$-regular,
its largest eigenvalue is $\lambda_0=d$, with the
uniform eigenvector. Moreover, $L=dI-\HM$ and
$\HM$ share the same eigenvectors.

Let $\vec{v}_k$ denote the real orthonormal
eigenvectors of $\HM$, satisfying
\[
\HM\vec{v}_k=\lambda_k\vec{v}_k.
\]
Then
\[
L\vec{v}_k
=
(\lambda_0-\lambda_k)\vec{v}_k.
\]
Expanding the cost vector in this eigenbasis,
\[
\vec{c}=\sum_k a_k\vec{v}_k,
\qquad
a_k=\vec{v}_k^{\top}\vec{c},
\]
and using orthonormality, we obtain
\begin{equation}
\begin{aligned}
\vec{c}^{\top}L\vec{c}
&=
\sum_k a_k^2(\lambda_0-\lambda_k)
\\
&=
\sum_{k\ne0}a_k^2(\lambda_0-\lambda_k),
\end{aligned}
\label{S-eq:step2}
\end{equation}
where the $k=0$ term vanishes because
$\lambda_0=d$.

We now relate the coefficients $a_k$ to the matrix
elements $C_{0k}$ defined earlier. Since
$\ket{\phi_0}=\ket{\psi_0}$ is the uniform state,
\[
\begin{aligned}
C_{0k}
&=
\bra{\phi_0}\HC\ket{\phi_k}
\\
&=
\frac{1}{\sqrt{|F|}}
\sum_{x\in F}C(x)v_k(x)
\\
&=
\frac{a_k}{\sqrt{|F|}}.
\end{aligned}
\]
Consequently,
\begin{equation}
a_k^2=|F|\,|C_{0k}|^2.
\label{S-eq:step3}
\end{equation}

Combining \eqref{S-eq:step1}--\eqref{S-eq:step3},
we obtain
\[
\boxed{
\mathcal E_M
=
|F|\sum_{k\ne0}
(\lambda_0-\lambda_k)|C_{0k}|^2.
}
\]
This proves \eqref{S-eq:equivalence-identity}.

Substituting this identity into
\eqref{S-eq:linear-bounds} gives the same
leading-order expressions as those in
\eqref{S-eq:gamma-bound}.
\end{proof}

\subsection{Proofs of the exact-objective bounds}
\label{S-supp:exact}
We use the optimized exact objective over the restricted parameter
box and the remainder operators
\[
M_p(\eta,\delta)=
\max_{\substack{0\leq\beta_i\leq\eta\\0\leq\gamma_i\leq\delta}}
\bra{\psi_p}\HC\ket{\psi_p},\qquad
R_\beta=\mathcal O^{(0)}-A_p,\quad
R_\gamma=\mathcal O^{(0)}-B_p.
\]
The derivation of their uniform norm estimates is given in
\Cref{S-supp:remainder} below.

\begin{restatedresult}{Corollary~\ref{cor:exact-beta} --- Exact small-$\beta$ objective bound}
Assume the conditions of Proposition~\ref{prop:linear-bounds} and
\[
2p\|\HM\|\eta\leq1.
\]
Define
\[
\varepsilon_\beta
=
4p^2\eta^2\|\HM\|^2\|\HC\|.
\]
Then
\[
\begin{aligned}
\bar C
+
\frac{2\eta\delta\,p(p+1)}{\pi|F|}
\mathcal{E}_M
-\varepsilon_\beta
&\leq M_p(\eta,\delta)
\\
&\leq
\bar C
+
\frac{\eta\delta\,p(p+1)}{|F|}
\mathcal{E}_M
+\varepsilon_\beta.
\end{aligned}
\]
\end{restatedresult}

\begin{proof}[Proof of Corollary~\ref{cor:exact-beta}]
For every parameter vector in the allowed range,
\[
\left|
\bra{\psi_p}\HC\ket{\psi_p}
-
\bra{\psi_0}A_p\ket{\psi_0}
\right|
\leq
\|R_\beta\|
\leq
\varepsilon_\beta.
\]
Therefore, the maximum of the exact objective differs from the maximum of its first-order approximation by at most $\varepsilon_\beta$. Combining this with Proposition~\ref{prop:linear-bounds} proves the result.
\end{proof}

\begin{restatedresult}{Corollary~\ref{cor:exact-gamma} --- Exact small-$\gamma$ objective bound}
Assume the conditions of Proposition~\ref{prop:gamma-bound} and
\[
2p\|\HC\|\delta\leq1.
\]
Define
\[
\varepsilon_\gamma
=
4p^2\delta^2\|\HC\|^3.
\]
Then
\[
\begin{aligned}
\bar C
+
\frac{2\delta\eta\,p(p+1)}{\pi}
\sum_{b\ne0}
(\lambda_0-\lambda_b)|C_{0b}|^2
-\varepsilon_\gamma
&\leq
M_p(\eta,\delta)
\\
&\leq
\bar C
+
\delta\eta\,p(p+1)
\sum_{b\ne0}
(\lambda_0-\lambda_b)|C_{0b}|^2
+\varepsilon_\gamma.
\end{aligned}
\]
\end{restatedresult}

\begin{proof}[Proof of Corollary~\ref{cor:exact-gamma}]
The argument is the same as in Corollary~\ref{cor:exact-beta}. The uniform error bound
\[
\left|
\bra{\psi_p}\HC\ket{\psi_p}
-
\bra{\psi_0}B_p\ket{\psi_0}
\right|
\leq
\|R_\gamma\|
\leq
\varepsilon_\gamma
\]
allows us to apply Proposition~\ref{prop:gamma-bound} to the exact objective.
\end{proof}

\subsection{Proofs of the spreading results}
\label{S-supp:spreading}
For each mixer layer define
\[
K_j(a,b)=\bra{a}e^{-i\beta_j\HM}\ket{b},\qquad
S_j=\max_{l\in F}\left\|e^{-i\beta_j\HM}\ket{l}\right\|_1,
\qquad \Lambda=\prod_{j=1}^{p} S_j^2.
\]
Here $\HM$ is real symmetric, so the mixer-unitary matrix
$K_j$ is symmetric, and the maximum row and column
$\ell_1$ norms coincide. Write
$q(x)=|\braket{x|\psi_p}|^2$ for the output probability.

\begin{restatedresult}{Lemma~\ref{lem:generalized} --- Path expansion of the QAOA objective}
For a depth-$p$ QAOA circuit, let
\[
q(x)=|\braket{x|\psi_p}|^2
\]
denote the probability of measuring a configuration $x\in F$. Then
\begin{equation}
\bra{\psi_p}\HC\ket{\psi_p}
=
\sum_{x\in F}C(x)q(x),
\label{S-eq:path-objective}
\end{equation}
where
\begin{equation}
\begin{aligned}
q(x)&=\frac{|A_x|^2}{|F|},\\
A_x
&=
\sum_{a_0,\ldots,a_{p-1}\in F}
\prod_{j=1}^{p}
K_j(a_j,a_{j-1})
e^{-i\gamma_jC(a_{j-1})},
\end{aligned}
\label{S-eq:path-sum}
\end{equation}
with $a_p=x$.
\end{restatedresult}

\begin{proof}[Proof of Lemma~\ref{lem:generalized}]
For $j=0,\ldots,p$, define $f_j(a)=\braket{a|\psi_j},$
where $\ket{\psi_j}$ is the state after $j$ QAOA
layers. 
Since $\ket{\psi_j}
=
U_M(\beta_j)U_C(\gamma_j)\ket{\psi_{j-1}},$
and $\HC$ is diagonal in the computational basis,
the amplitudes satisfy the recurrence
\begin{equation}
f_j(a_j)
=
\sum_{a_{j-1}\in F}
K_j(a_j,a_{j-1})
e^{-i\gamma_j C(a_{j-1})}
f_{j-1}(a_{j-1}).
\label{S-eq:path-recurrence}
\end{equation}
Moreover, the uniform initial state gives
$f_0(a_0)=1/\sqrt{|F|}$ for every $a_0\in F$.

We show by induction on $j$ that
\[
\begin{aligned}
f_j(a_j)
=
\frac{1}{\sqrt{|F|}}
\sum_{a_0,\ldots,a_{j-1}\in F}
\prod_{\ell=1}^{j}
K_\ell(a_\ell,a_{\ell-1})
e^{-i\gamma_\ell C(a_{\ell-1})}.
\end{aligned}
\]
For $j=1$, the expression follows directly from
\eqref{S-eq:path-recurrence} and the uniform initial
amplitudes. Suppose the expression holds for $j-1$.
Substituting the induction hypothesis into
\eqref{S-eq:path-recurrence}, we obtain
\[
\begin{aligned}
f_j(a_j)
={}&
\frac{1}{\sqrt{|F|}}
\sum_{a_{j-1}\in F}
K_j(a_j,a_{j-1})
e^{-i\gamma_j C(a_{j-1})}
\\
&\quad\times
\sum_{a_0,\ldots,a_{j-2}\in F}
\prod_{\ell=1}^{j-1}
K_\ell(a_\ell,a_{\ell-1})
e^{-i\gamma_\ell C(a_{\ell-1})}
\\
={}&
\frac{1}{\sqrt{|F|}}
\sum_{a_0,\ldots,a_{j-1}\in F}
\prod_{\ell=1}^{j}
K_\ell(a_\ell,a_{\ell-1})
e^{-i\gamma_\ell C(a_{\ell-1})}.
\end{aligned}
\]
Thus, the claimed expansion holds for every
$j=1,\ldots,p$.

Setting $j=p$ and $a_p=x$ gives
\[
\braket{x|\psi_p}
=
\frac{A_x}{\sqrt{|F|}},
\]
where $A_x$ is defined in \eqref{S-eq:path-sum}.
Consequently,
\[
q(x)
=
|\braket{x|\psi_p}|^2
=
\frac{|A_x|^2}{|F|}.
\]
Finally, since $\HC$ is diagonal in the computational
basis,
\[
\bra{\psi_p}\HC\ket{\psi_p}
=
\sum_{x\in F}C(x)q(x),
\]
which establishes the result.
\end{proof}

\begin{restatedresult}{Lemma~\ref{lem:pointwise} --- Upper bound on output probability}
For every $x\in F$, the probability of measuring $x$ at the end of a depth-$p$ QAOA circuit satisfies
\[
q(x)\leq\frac{\Lambda}{|F|}.
\]
\end{restatedresult}

\begin{proof}[Proof of Lemma~\ref{lem:pointwise}]
From \eqref{S-eq:path-sum}, every cost-dependent phase has unit magnitude,
\[
\left|e^{-i\gamma_jC(a_{j-1})}\right|=1.
\]
Applying the triangle inequality therefore gives
\[
|A_x|
\leq
\sum_{a_0,\ldots,a_{p-1}\in F}
\prod_{j=1}^{p}
|K_j(a_j,a_{j-1})|.
\]

Recall that $S_j = \max_{l\in F}
\left\|
e^{-i\beta_j\HM}\ket{l}
\right\|_1$. 
Since $\HM$ is real symmetric, the mixer unitary
is symmetric, so $K_j(a,b)=K_j(b,a)$.
Consequently, for any fixed $a_j\in F$,
\[
\begin{aligned}
\sum_{a_{j-1}\in F}|K_j(a_j,a_{j-1})|
&=
\sum_{a_{j-1}\in F}|K_j(a_{j-1},a_j)|
\\
&=
\left\|e^{-i\beta_j\HM}\ket{a_j}\right\|_1
\\
&\leq S_j.
\end{aligned}
\]

We can therefore bound the sum successively over $a_0,a_1,\ldots,a_{p-1}$, with each summation bounded by the corresponding row norm. Each summation introduces a factor of at most $S_j$, giving
\[
|A_x|\leq\prod_{j=1}^{p}S_j.
\]

Taking the squared magnitude and using the definition of $\Lambda$, we obtain
\[
|A_x|^2
\leq
\prod_{j=1}^{p}S_j^2
=
\Lambda.
\]
Since $q(x)=|A_x|^2/|F|$, it follows that
\[
q(x)\leq\frac{\Lambda}{|F|},
\]
as required.
\end{proof}

\begin{restatedresult}{Theorem~\ref{thm:bathtub} --- Objective bound from mixer spreading}
Let $N=|F|$, and arrange the objective values in decreasing order:
\[
\Cmax=C_{(1)}
\geq C_{(2)}
\geq\cdots
\geq C_{(N)}.
\]
Define
\[
\kappa=\frac{N}{\Lambda},
\qquad
m=\lfloor\kappa\rfloor.
\]
Then the expected QAOA objective satisfies
\[
\bra{\psi_p}\HC\ket{\psi_p}
\leq B(\Lambda),
\]
where
\begin{equation}
B(\Lambda)
=
\begin{cases}
\Cavg,
& \Lambda=1,
\\[2mm]
\dfrac{\Lambda}{N}
\left(
\displaystyle\sum_{i=1}^{m}C_{(i)}
+
(\kappa-m)C_{(m+1)}
\right),
& 1<\Lambda<N,
\\[3mm]
\Cmax,
& \Lambda\geq N.
\end{cases}
\label{S-eq:spreading-objective-bound}
\end{equation}

Here, $\Cavg$ denotes the average objective value
over all feasible configurations. The bound
satisfies
\[
B(1)=\Cavg,
\qquad
B(\Lambda)\leq\Cmax,
\]
and is non-decreasing in $\Lambda$.
Furthermore, if $C(x)\geq0$ for every $x\in F$,
then
\[
B(\Lambda)\leq\Lambda\,\Cavg.
\]
\end{restatedresult}

\begin{proof}[Proof of Theorem~\ref{thm:bathtub}]
By Lemma~\ref{lem:pointwise}, every output probability satisfies
\[
0\leq q(x)\leq\frac{\Lambda}{N},
\qquad
\sum_{x\in F}q(x)=1.
\]
The expected QAOA objective is therefore bounded above by the maximum of the linear program
\begin{equation}
\begin{aligned}
\max_q \quad&
\sum_{x\in F}C(x)q(x)
\\
\text{subject to}\quad&
0\leq q(x)\leq\frac{\Lambda}{N},
\\
&
\sum_{x\in F}q(x)=1.
\end{aligned}
\label{S-eq:spreading-lp}
\end{equation}

To characterize the optimum, consider two configurations with $C(x)>C(y)$. If $q(y)>0$ while $q(x)<\Lambda/N$, transferring some probability from $y$ to $x$ increases the objective without violating the constraints. Consequently, an optimal distribution can be constructed by assigning the maximum allowed probability to the highest-value configurations before assigning probability to lower-value ones. We now evaluate the resulting objective in the different regimes of $\Lambda$.

When $\Lambda=1$, the probability constraint becomes $q(x)\leq1/N$ for every configuration. Since the probabilities must sum to one, the only feasible distribution is uniform. Therefore,
\[
B(1)=\frac{1}{N}\sum_{x\in F}C(x)=\Cavg.
\]

For $1<\Lambda<N$, let
\[
\kappa=\frac{N}{\Lambda},
\qquad
m=\lfloor\kappa\rfloor.
\]
The first $m$ configurations receive the maximum allowed probability,
\[
q_{(i)}=\frac{\Lambda}{N},
\qquad i=1,\ldots,m.
\]
The remaining probability is
\[
1-m\frac{\Lambda}{N}
=
(\kappa-m)\frac{\Lambda}{N},
\]
which is assigned to the configuration with the next-largest objective value. All subsequent configurations receive zero probability, and the remaining contribution vanishes when $\kappa$ is an integer. The corresponding objective is
\[
\begin{aligned}
\sum_{x\in F}C(x)q(x)
&=
\frac{\Lambda}{N}
\sum_{i=1}^{m}C_{(i)}
\\
&\quad+
(\kappa-m)\frac{\Lambda}{N}
C_{(m+1)}
\\
&=B(\Lambda).
\end{aligned}
\]
This distribution satisfies all the probability constraints and is optimal by the preceding exchange argument.

Finally, if $\Lambda\geq N$, the pointwise bound satisfies $\Lambda/N\geq1$ and imposes no restriction beyond the normalization of probabilities. A distribution concentrated entirely on a configuration attaining $\Cmax$ is then feasible for the linear program, giving $B(\Lambda)=\Cmax$. Moreover, increasing $\Lambda$ relaxes the probability constraints, so the optimal value $B(\Lambda)$ is non-decreasing. Since the objective is an average of the values $C(x)$, it cannot exceed $\Cmax$.

For nonnegative objective functions, the pointwise probability bound also gives
\[
\begin{aligned}
\sum_{x\in F}C(x)q(x)
&\leq
\frac{\Lambda}{N}
\sum_{x\in F}C(x)
\\
&=
\Lambda\,\Cavg.
\end{aligned}
\]
\end{proof}

\section{Additional Estimates and Numerical Controls}
\label{S-supp:additional}
\subsection{Detailed truncation-error derivation}
\label{S-supp:remainder}
For $X\in\{\HM,\HC\}$ and a list of angles
$\theta=(\theta_1,\ldots,\theta_p)$, the second-order
remainder of the corresponding Heisenberg expansion satisfies
\[
\|R\|\leq 4\|X\|^2\left(\sum_{\ell=1}^p|\theta_\ell|\right)^2\|\HC\|
\quad\text{if }2\|X\|\sum_{\ell=1}^p|\theta_\ell|\leq1.
\]
In particular, this gives the $R_\beta$ and $R_\gamma$
estimates used in Corollaries~\ref{cor:exact-beta} and~\ref{cor:exact-gamma}. We provide the
complete derivation below.

We use the standard nested-commutator bound
\[
\|\com_X^n(Y)\|
\leq
(2\|X\|)^n\|Y\|,
\]
where $\|\cdot\|$ denotes the operator norm.
For either expansion, the discarded terms have total order at least two in the corresponding small parameter. Their sum satisfies
\begin{equation}
\|R\|
\leq
\sum_{\substack{n_1,\ldots,n_p\geq0\\
\sum_\ell n_\ell\geq2}}
\left(
\prod_{\ell=1}^{p}
\frac{|\theta_\ell|^{n_\ell}}{n_\ell!}
\right)
(2\|X\|)^{\sum_\ell n_\ell}
\|\HC\|,
\label{S-eq:error-step1}
\end{equation}
where
\[
(X,\theta,R)=
\begin{cases}
(\HM,\beta,R_\beta), & \text{small-}\beta,\\
(\HC,\gamma,R_\gamma), & \text{small-}\gamma.
\end{cases}
\]

Here we use the fact that conjugation by a unitary does not change the operator norm.

The sum in \eqref{S-eq:error-step1} is the part of an exponential series starting at second order. Hence,
\begin{equation}
\|R\|
\leq
\left(
e^{\,2\|X\|\sum_\ell|\theta_\ell|}
-1
-2\|X\|\sum_\ell|\theta_\ell|
\right)
\|\HC\|.
\label{S-eq:error-step2}
\end{equation}

For
\[
0\leq
2\|X\|\sum_{\ell=1}^{p}|\theta_\ell|
\leq1,
\]
we can use
\[
e^t-1-t\leq t^2
\qquad (0\leq t\leq1)
\]
to obtain
\begin{equation}
\|R\|
\leq
4\|X\|^2
\left(
\sum_{\ell=1}^{p}|\theta_\ell|
\right)^2
\|\HC\|.
\label{S-eq:error-step3}
\end{equation}

Applying this bound separately to the two expansions gives
\begin{equation}
\begin{aligned}
\|R_\beta\|
&\leq
4\|\HM\|^2
\left(\sum_{\ell=1}^{p}|\beta_\ell|\right)^2
\|\HC\|,
\\
\|R_\gamma\|
&\leq
4\|\HC\|^2
\left(\sum_{\ell=1}^{p}|\gamma_\ell|\right)^2
\|\HC\|.
\end{aligned}
\label{S-eq:supp-both-error-bounds}
\end{equation}

The two tolerance conditions stated in the main article follow
immediately from the final display.
\begin{restatedresult}{Corollary~\ref{cor:eta-bound} --- Small-$\beta$ tolerance condition}
If $2p\|\HM\|\eta\leq1$, a sufficient condition for
$\|R_\beta\|\leq\epsilon$ is
\[
\eta\leq\frac{1}{2p\|\HM\|}\sqrt{\frac{\epsilon}{\|\HC\|}}.
\]
\end{restatedresult}
\begin{restatedresult}{Corollary~\ref{cor:delta-bound} --- Small-$\gamma$ tolerance condition}
If $2p\|\HC\|\delta\leq1$, a sufficient condition for
$\|R_\gamma\|\leq\epsilon$ is
\[
\delta\leq\frac{1}{2p\|\HC\|}\sqrt{\frac{\epsilon}{\|\HC\|}}.
\]
\end{restatedresult}
Both claims follow by substituting the uniform upper bounds on
$\sum_j|\beta_j|$ and $\sum_j|\gamma_j|$ into the preceding
remainder estimates. These expressions presume the nonzero norms
appearing in their denominators; degenerate cases can be handled
directly.

\subsection{Exploratory parameter-budget control}
\label{S-supp:param-control}
The primary comparison assigns exactly one additional shared variational parameter per layer to each augmented ansatz. Standard $X$ uses the two parameters $(\gamma_\ell,\beta_\ell)$, while the augmented mixers use $(\gamma_\ell,\beta_\ell,\beta'_\ell)$. This keeps the additional parameter cost modest while preserving a globally consistent circuit in which every edge interaction of the same type shares its angle. Nevertheless, an improvement at equal depth could partly reflect the greater flexibility provided by the additional parameter.

To investigate this possibility, we also considered an exploratory control based on the standard $X$ mixer with two independent $\beta$ angles per layer. The qubits within each representative lightcone were divided into two groups, with a separate transverse-field angle assigned to each group. This introduces an additional local parameter without adding two-qubit mixer interactions, allowing us to examine whether parameter count alone could account for the observed improvement. In these tests, the additional parameter produced negligible improvement over the standard single-$\beta$ baseline.

This control has an important limitation: the grouping was defined separately within each representative lightcone and does not constitute a globally consistent ansatz on the original graph. Consequently, it cannot serve as a strict parameter-matched global baseline, and its results should be interpreted only as an exploratory comparison. Nevertheless, the limited improvement suggests that the gains observed for the augmented mixers are unlikely to be explained solely by adding one variational parameter. A globally defined parameter-matched control would be needed to test this explanation more conclusively.

\subsection{Additional mixer variants}
\label{S-supp:additional-mixers}
We also evaluated several additional mixer constructions to examine whether the observed behavior is specific to the $XX$ and $YY$ augmentations. The standalone edge mixers considered were $XX$, $YY$, $ZZ$, and
\[
AA(\theta)
=
(\cos\theta X+\sin\theta Y)_i
(\cos\theta X+\sin\theta Y)_j,
\]
where $\theta$ is a shared, jointly optimized angle. We additionally considered the augmented variants $X+ZZ$ and $X+AA(\theta)$. These experiments provide further comparisons between mixer interactions with different algebraic properties.

Two of the standalone mixers can be shown to produce no improvement over the uniform-state objective. For standalone $ZZ$, both the cost and mixer Hamiltonians are diagonal in the computational basis, so the circuit changes only the phases of the initial uniform superposition. The measurement distribution therefore remains uniform, and the expected objective is exactly $\bar C$ for every choice of parameters. Standalone $YY$ also has a symmetry that makes its edge-cost expectation equal to $1/2$ on the bipartite lightcones considered here.

To see the latter property, let $A$ denote one side of the bipartition and define the antiunitary operator
\[
\Theta
=
K\prod_{a\in A}X_a,
\]
where $K$ denotes complex conjugation in the computational basis. For every edge $(u,v)$ joining the two parts of the bipartition,
\[
\Theta Z_uZ_v\Theta^{-1}
=
-Z_uZ_v,
\qquad
\Theta Y_uY_v\Theta^{-1}
=
-Y_uY_v.
\]
Because $\Theta$ is antiunitary, these sign changes preserve the corresponding cost and standalone-$YY$ unitaries up to an irrelevant global phase. The uniform initial state is also invariant under $\Theta$, so the evolved state has the same symmetry. It follows that
$\langle Z_uZ_v\rangle=0$ and hence
$\langle C_{uv}\rangle_p=1/2$ for every edge, independently of the circuit parameters.

The standalone $XX$ and $AA(\theta)$ mixers do not share this degeneracy. Both nevertheless preserve Hamming-weight parity: their edge interactions flip either zero or two computational-basis bits. In our numerical experiments, neither construction improved on the standard $X$ baseline. This indicates that simply introducing two-qubit transitions is not sufficient to obtain the improvements observed with the augmented mixers. The interaction must also combine effectively with the transverse-field component and the cost evolution.

Among the augmented variants, $X+ZZ$ reproduces the standard-$X$ optimum. For MaxCut, the $ZZ$ edge operator is an affine function of the cost Hamiltonian, and the corresponding diagonal mixer block can be combined with an adjacent cost layer by shifting its angle. Any remaining terminal diagonal block does not affect the expectation of the diagonal cost observable. Thus, the additional $ZZ$ parameter does not enlarge the set of achievable objective values relative to the standard ansatz.

The $X+AA(\theta)$ mixer provides a continuous interpolation between the $X+XX$ and $X+YY$ constructions, corresponding to $\theta=0$ and $\theta=\pi/2$, respectively. In every setting studied, its numerical optimization converged to the best-found value obtained with $X+YY$, without producing a further improvement from the additional angle $\theta$. While this observation does not establish global optimality over the enlarged mixer family, it suggests that the $YY$ interaction is particularly effective within the variants examined. The result is also consistent with the reach-matched comparisons, where $X+YY$ performed substantially better than $X+XX$ despite having the same lightcone-radius bound.

These observations point to mixer operator structure as an important consideration alongside causal reach. The $XX$ augmentation commutes with the transverse-field mixer, so the combined $X+XX$ mixer remains diagonal in the common $X$ eigenbasis. In contrast, the $YY$ edge interactions do not commute with the transverse-field $X$ rotations, allowing the ordered combination to generate operator structures that are absent from either component alone. The fact that standalone $YY$ is ineffective while $X+YY$ gives the strongest best-found results among the tested mixers illustrates this distinction. Although the experiments do not isolate a unique mechanism responsible for the improvement, they show that the usefulness of an augmentation depends on how its interactions combine with the rest of the ansatz, not merely on how far its lightcone extends.

\section{Supporting Proofs for Causal Reach}
\label{S-supp:locality}
\subsection*{MaxCut circuits and $k$-hop mixer augmentation}

Let $G=(V,E)$ be an undirected graph, one qubit per vertex,
and let the MaxCut cost Hamiltonian be
\[
\HC=\sum_{(u,v)\in E}C_{uv},\qquad
C_{uv}=\frac{I-Z_uZ_v}{2}.
\]
The standard transverse-field mixer is
$\HMk{0}=\sum_{v\in V}X_v$.
For an integer $k\geq1$, define
\[
\Eaug{k}=\bigl\{\{i,j\}\subseteq V:\;i\neq j,
\;\dist_G(i,j)\leq k\bigr\}.
\]
For $P\in\{X,Y\}$ the augmented mixer layer is the ordered
unitary
\[
U_{M,\ell}^{(k,P)}=
\left(\prod_{\{i,j\}\in\Eaug{k}}e^{-i\beta'_\ell P_iP_j}\right)
\left(\prod_{v\in V}e^{-i\beta_\ell X_v}\right).
\]
For $P=X$ and $P=Y$ these are the $X{+}XX$ and $X{+}YY$
augmentations. The two-qubit gates within each $P_iP_j$ block
commute; the single-qubit rotations need not commute with that
block when $P=Y$. The depth-$p$ QAOA unitary $U_p$ consists of
$p$ alternating cost and chosen mixer layers. The support radius
of an operator is measured in graph distance from the endpoints
of the marked edge $(u,v)$.

\subsection{Proofs of the causal-reach bounds}
\label{S-supp:lightcone}

\begin{restatedresult}{Lemma~\ref{lem:lightcone} --- Lightcone of standard QAOA}
Consider a depth-$p$ QAOA circuit with the standard transverse-field mixer $\HMk{0}$. For any edge $(u,v)\in E$, the Heisenberg-evolved operator $U_p^\dagger C_{uv}U_p$ is supported only on qubits whose graph distance from $u$ or $v$ is at most $p$. More precisely, its support is contained in
\[
L_p(u,v)
=
\left\{
w\in V:
\dist_G(w,u)\leq p
\;\text{or}\;
\dist_G(w,v)\leq p
\right\},
\]
where $U_p$ denotes the depth-$p$ QAOA unitary. Hence, $\langle C_{uv}\rangle_p$ can be evaluated using the interactions contained in this local neighborhood.
\end{restatedresult}

\begin{proof}
We use the standard operator-support argument. For a single QAOA layer, the Heisenberg-evolved observable is
\[
U_1^\dagger C_{uv}U_1
=
e^{i\gamma_1\HC}
e^{i\beta_1\HMk{0}}
C_{uv}
e^{-i\beta_1\HMk{0}}
e^{-i\gamma_1\HC}.
\]
The transverse-field mixer consists of single-qubit rotations, and rotations acting on qubits other than $u$ and $v$ commute with $C_{uv}$. These rotations therefore cancel in the conjugation, so the mixer does not enlarge the support of the observable.

The cost unitary consists of commuting $ZZ$ interactions along the edges of $G$. Only interactions that overlap with the current support can affect the observable, and these interactions can introduce qubits adjacent to $u$ or $v$. Thus, after one layer, the evolved operator is supported within graph distance one of $\{u,v\}$. The same reasoning extends to additional layers: single-qubit mixer rotations preserve the current support, while each cost layer can enlarge it by at most one graph hop. After $p$ layers, the support is therefore contained in $L_p(u,v)$, proving the result.
\end{proof}

\begin{restatedresult}{Lemma~\ref{lem:lightcone-growth} --- Lightcone growth under $k$-hop augmentation}
Consider a depth-$p$ QAOA circuit with the MaxCut cost Hamiltonian $\HC$ and either of the $k$-hop augmented mixer circuits from the $k$-hop mixer construction defined above. For any edge $(u,v)\in E$, the Heisenberg-evolved cost operator is supported within graph distance at most
\[
r(p)=(k+1)p
\]
of $\{u,v\}$. For the standard transverse-field mixer, which has no two-qubit augmentation, the corresponding radius bound is $r(p)=p$.
\end{restatedresult}

\begin{proof}[Proof of Lemma~\ref{lem:lightcone-growth}]
We examine how the support of a local operator can grow under each part of a QAOA layer. The cost layer consists of commuting $ZZ$ interactions on the edges of $G$, so it can enlarge the support by at most one graph hop, as in the standard-mixer argument of Lemma~\ref{lem:lightcone}. The augmented mixer contains single-qubit $X$ rotations and two-qubit $P_iP_j$ rotations, where $P\in\{X,Y\}$ and $\dist_G(i,j)\leq k$. The single-qubit rotations do not enlarge the support, but the two-qubit rotations can introduce qubits at most $k$ graph hops away from the current support.

For a fixed choice of $P$, all the two-qubit $P_iP_j$ generators commute with one another. Moreover, a qubit introduced into the support through one such interaction carries a Pauli $P$ operator, which commutes with the other two-qubit generators involving that qubit. Thus, the two-qubit interaction block cannot repeatedly extend the support through successive augmentation edges within the same block. Its contribution to the growth of the support is therefore limited to at most $k$ graph hops.

Combining the mixer and cost layers, one complete QAOA layer can enlarge the support by at most $k+1$ graph hops. Repeating the argument over $p$ layers gives
\[
r(p)\leq(k+1)p.
\]
For the standard transverse-field mixer, no two-qubit interactions are present, so only the cost layers enlarge the support. This recovers the bound $r(p)\leq p$.
\end{proof}

\end{document}